\documentclass[runningheads,a4paper]{article}

\usepackage{a4wide}

\usepackage[utf8]{inputenc}
\usepackage[T1]{fontenc}   
\usepackage[english]{babel} 
\usepackage{amsmath}
\usepackage{amsfonts}
\usepackage{amssymb}
\usepackage{amsthm}
\usepackage{mathrsfs}
\usepackage{calrsfs}
\usepackage[mathscr]{eucal}

\usepackage{url}
\usepackage{enumerate}
\usepackage{graphicx}
\usepackage[all]{xy}
\usepackage{color}
\usepackage{enumerate}
\usepackage{booktabs}
\usepackage{tikz}
\usepackage{algorithm}
\usepackage{caption}
\usepackage{algpseudocode}
\usepackage{algorithm}

\usepackage{stmaryrd}

\makeatletter
\newcounter{restate}

\makeatother

\newtheorem{thm}{Theorem}
\newtheorem{theorem}{Theorem}

\newtheorem{proposition}[thm]{Proposition}
\newtheorem{problem}{Problem}
\newtheorem{corollary}[thm]{Corollary}

\newtheorem{lemma}[thm]{Lemma}
\newtheorem{definition}{Definition}

\newtheorem{remark}{Remark}
\usepackage{thm-restate}

\newtheorem{fact}{Fact}
\newtheorem{ass}{Assumption}

\newcommand{\cC}{\mathcal{C}}
\newcommand{\cD}{\mathcal{D}}

\newcommand{\cH}{\mathcal{H}}

\newcommand{\cS}{\mathcal{S}}

\newcommand{\cv}{{\mathbf{c}}}
\newcommand{\ev}{{\mathbf{e}}}

\newcommand{\mv}{{\mathbf{m}}}

\newcommand{\sv}{{\mathbf{s}}}

\newcommand{\uv}{{\mathbf{u}}}

\newcommand{\xv}{{\mathbf{x}}}
\newcommand{\yv}{{\mathbf{y}}}

\newcommand{\Am}{{\mathbf{A}}}

\newcommand{\Em}{{\mathbf{E}}}
\newcommand{\Gm}{{\mathbf{G}}}
\newcommand{\Hm}{{\mathbf{H}}}

\newcommand{\Mm}{{\mathbf{M}}}
\newcommand{\Pm}{{\mathbf{P}}}
\newcommand{\Qm}{{\mathbf{Q}}}
\newcommand{\Rm}{{\mathbf{R}}}
\newcommand{\Sm}{{\mathbf{S}}}

\newcommand{\Xm}{{\mathbf{X}}}
\newcommand{\Xv}{{\mathbf{X}}}
\newcommand{\Yv}{{\mathbf{Y}}}
\newcommand{\mat}[1]{\ensuremath{\boldsymbol{#1}}}
\newcommand{\un}{{\mat{1}}}

\newcommand{\Cc}{{\mathcal C}}

\newcommand{\Hc}{{\mathcal H}}

\newcommand{\Cpub}{\Cc_{\textup{pub}}}
\newcommand{\Cpubp}{\Cc_{\textup{pub}}^{\textup{proj}}}
\newcommand{\Cmat}{\cC^{\textup{Mat}}}
\newcommand{\Clrpc}{\Cc_{\textup{LRPC}}}
\newcommand{\clrpc}{\cv_{\textup{LRPC}}}

\newcommand{\Hpub}{{\Hm_{\textup{pub}}}}
\newcommand{\Hpubp}{{\Hm_{\textup{pub}}^{\textup{proj}}}}
\newcommand{\Hsec}{{\Hm_{\textup{sec}}}}

\newcommand{\wrGV}{w_{\textup{rVG}}}
\newcommand{\wrVG}{w_{\textup{rVG}}}
\newcommand{\ws}{w_{\textup{rS}}}

\newcommand{\F}{\mathbb{F}}

\newcommand{\Fq}{\F_q}

\newcommand{\Fqm}{\F_{q^m}}

\newcommand{\fqm}{\F_{q^m}}

\newcommand{\vsg}[1]{\langle #1 \rangle_{\Fq}}
\newcommand{\vsgm}[1]{\langle #1 \rangle_{\Fqm}}
\newcommand{\transpose}[1]{{#1}^{ {\intercal} }}
\newcommand{\transp}[1]{{#1}^{ {\intercal} }}
\newcommand{\defeq}{\eqdef}

\newcommand{\sgn}{\textup{sgn}}
\newcommand{\dec}{\textup{dec}}

\newcommand{\eqdef}{\mathop{=}\limits^{\triangle}}

\DeclareMathOperator*{\Sp}{Supp}
\DeclareMathOperator*{\Sup}{Supp}

\newcommand{\Mat}[1]{\textup{\textbf{Mat}}(#1)}
\newcommand{\Matp}[1]{\textup{\textbf{Mat}}^{\textup{proj}}(#1)}

\newcommand{\IInt}[2]{\llbracket #1, #2 \rrbracket}

\begin{document}
	\title{Two attacks on rank metric code-based schemes: RankSign and an Identity-Based-Encryption
	scheme}

\author{Thomas Debris-Alazard  \footnote{Sorbonne Universit\'{e}s, UPMC Univ Paris 06, France} \footnote{ Inria, SECRET Project, 2 Rue Simone Iff 75012 Paris Cedex} \and Jean-Pierre Tillich $^{\dagger}$ }

	\maketitle
	\begin{abstract}
RankSign \cite{GRSZ14a} is a code-based signature scheme proposed to the NIST competition for quantum-safe cryptography
\cite{AGHRZ17} and, moreover, is a fundamental building block of a new Identity-Based-Encryption (IBE) \cite{GHPT17a}. This signature scheme is based on the rank metric and enjoys remarkably small key sizes, about 10KBytes for an intended level of security of 128 bits. 
Unfortunately we will show that all the parameters proposed for this scheme in \cite{AGHRZ17} can be broken by an algebraic attack that exploits the fact that the augmented LRPC codes
used in this scheme have very low weight codewords. 
Therefore, without RankSign the IBE cannot be instantiated at this time. As a second contribution we will show that the problem is deeper than finding a new signature in rank-based cryptography, 
we also found an attack on the generic problem upon which its security reduction relies. However, contrarily to the RankSign scheme, it seems that the parameters of the
IBE scheme could be chosen in order to avoid our attack. Finally, we have also shown that if one replaces the rank metric in the 
\cite{GHPT17a} IBE scheme by the Hamming metric, then a devastating attack can be found.
	\end{abstract}

	\section{Introduction}
	
	\subsection{An efficient code-based signature scheme: RankSign and a code-based Identity-Based-Encryption scheme}
	
	{\bf Code-based signature schemes.}
	It is a long standing open problem to build an efficient and secure signature scheme based on the
	hardness of decoding a linear code which could compete in all respects
	with DSA or RSA. Such schemes could indeed give a quantum resistant 
	signature for replacing in
	practice the aforementioned signature schemes  that are well known to be broken by  
	quantum computers. A first partial answer to this question was given in
	\cite{CFS01}. It consisted in adapting the Niederreiter scheme \cite{N86} for
	this purpose. This requires a linear code for which there exists an efficient  decoding algorithm for 
	a non-negligible set of inputs.
	This means that if $\Hm$ is an
	$r \times n$ parity-check matrix of the code,
	there exists for a non-negligible set of elements $\sv$ in  $\{0,1\}^r$ an efficient way to
	find a word $\ev$ in $\{0,1\}^n$ of smallest Hamming weight such that $\Hm
	\transp{\ev}=\transp{\sv}$. 
	
	The authors of \cite{CFS01} noticed that very high rate Goppa codes are
	able to fulfill this task, and their scheme can indeed be
	considered
	as the first step towards a solution of the aforementioned problem.
	However, the poor scaling of the key size  when security has to be increased prevents 
	this scheme to be a completely satisfying answer to this issue.
	
	\noindent {\bf The rank metric.}
	There has been some exciting progress in this area for another metric, namely the rank metric \cite{GRSZ14a}.
	A  code-based signature scheme whose security relies 
	on decoding codes with respect to the rank metric has been proposed there.
	It is called RankSign.
	Strictly speaking, the rank metric consists in viewing an element in $\Fq^N$ (when $N$ 
	is a product $N= m \times n$) as an $m \times n$ matrix over $\Fq$ and the rank distance between 
	two elements $\xv$ and $\yv$ is defined as the rank of the matrix $\xv - \yv$. This depends of course on
	how $N$ is viewed as a product of two elements. Decoding in this metric is known to be an NP hard problem 
	\cite{BFS99,C01}.
	In the particular case of \cite{GRSZ14a}, the codes which are considered are not $\Fq$-linear but, as is customary
	in the setting of rank metric based cryptography, $\Fqm$-linear: the codes are here subspaces of $\Fqm^n$.
	Here the
	elements $\xv=(x_1,\dots,x_n)$ of $\Fqm^n$ are viewed as $m \times n$ matrices by expressing 
	each coordinate $x_i$ in a certain fixed $\Fq$-basis of $\Fqm$. This yields a column vector $\xv^i$ in $\Fq^m$ and 
	the concatenation of these column vectors yields an $m \times n$ matrix $\Mat{\xv}=\begin{pmatrix}
	\xv^1& \hdots &\xv^n\end{pmatrix}$ that allows to put a rank metric over $\Fqm^n$.
	This allows to reduce the key size by a factor of $m$ when compared to the 
	$\Fq$-linear setting (for more details see the paragraph at the end of Section \ref{sec:notation}).

	Decoding such codes for the rank metric is not known to be NP-hard anymore. There is 
	however a randomized reduction of this problem to decode an  $\Fq$-linear code for the Hamming metric \cite{GZ14} when the 
	degree $m$ of the extension field is sufficiently big. This situation is in some sense reminiscent to the current thread in 
	cryptography based on codes or on lattices where structured codes (for instance quasi-cyclic codes) or structured 
	lattices (corresponding to an additional ring structure) are taken. However the $\Fqm$-linear case has an advantage over the other 
	structured proposals, in the sense that it has a randomized reduction to an NP complete problem. 
	This is not the case for the other structured proposals. 
	Relying on $\Fqm$-linear codes is one of the main reason why RankSign enjoys noticeably small public key sizes: 
	it is about 10KBytes for 128 bits of security for the parameters proposed in the NIST submission\cite{AGHRZ17}. Furthermore, RankSign comes with a security proof showing that there is no leakage coming from signing many times.
	It also proved to be a fundamental building block in the Identity-Based-Encryption (IBE) scheme based on the rank metric suggested in \cite{GHPT17a}. 
	
	\noindent
	{\bf A new IBE scheme based on codes.}
	The concept of IBE was introduced by Shamir in 1984 \cite{S84}. It gives an alternative to the standard notion of public-key encryption. In an IBE scheme, the public key associated with a user can be an arbitrary identity string, such as his e-mail address, and others can send encrypted messages to a user using his identity without having to rely on a public-key infrastructure, given short public parameters. The main technical difference between a Public Key Encryption (PKE) and  IBE is the way  the public and private keys are bound and the way of verifying those keys. In a PKE scheme, verification is achieved through the use of a certificate which relies on a public-key infrastructure. In an IBE, there is no need of verification of the public key but the private key is managed by a Trusted Authority (TA).
	
	There are two issues that makes the design of IBE extremely hard: the requirement that public keys 
	are arbitrary strings and the ability to extract decryption keys from the public keys. 
	In fact, it took nearly twenty years for
	the problem of designing an efficient method to implement an IBE to be solved. 
	The known methods of designing IBE are based on different tools: from elliptic curve
	pairings \cite{SOK00} and  \cite{BF01}; from the quadratic residue problem \cite{C01}; from the Learning-With-Error (LWE) problem  \cite{GPV08}; from the computational Diffie-Hellman assumption \cite{DG17} and finally from the Rank Support Learning (RSL) problem \cite{GHPT17a}. The last scheme based on codes is an adaptation of the \cite{GPV08} technique, but instead of relying on the Hamming metric it relies on the rank metric. It has to be noted that there has been some recent and exciting progress in the design of IBE. In \cite{DG17a} it has been shown how to generalize the work of \cite{DG17} by introducing a new primitive, One-Time Signatures with Encryption (OTSE), that enables to construct fully secure IBE schemes. Furthermore it was shown in \cite{DGHM18} how to instantiate OTSE primitives from LWE and the Low Parity Noise problems (LPN). This gave after the IBE's \cite{GPV08} and \cite{GHPT17a} the third scheme which may hope to resist to a quantum computer.

	\subsection{Our contribution}
	
	{\bf An efficient attack on RankSign.} Our first contribution is that despite the fact that the security of  RankSign 
	might very well be founded on a hard problem (namely distinguishing an augmented LRPC code from a random
	linear code), we show here that all the parameters proposed for RankSign in \cite{AGHRZ17} can be broken by a suitable algebraic attack. The problem is actually deeper than that, because  the attack is actually polynomial in nature and can not really be thwarted by changing the parameters. The 
	attack builds upon the following observations
	\begin{itemize}
		\item The RankSign scheme is based on augmented LRPC codes;
		\item To have an efficient signature scheme, the parameters of the augmented LRPC codes
		have to be chosen very carefully;
		\item For the whole range of admissible parameters,  it turns out rather unexpectedly that these augmented LRPC codes have very low-weight codewords. This can be proved by subspace product considerations;
		\item These low-weight codewords can be recovered by algebraic techniques and reveal enough of the secret trapdoor
		used in the scheme to be able to sign like a legitimate user.
	\end{itemize}
	
	This attack has also a significant impact on the IBE proposal \cite{GHPT17a} whose security is based on the security of RankSign. Right now, there is no backup solution for instantiating this IBE scheme, since RankSign was the only rank-metric code based signature scheme following the hash
	and sign paradigm that is needed in the IBE scheme. 
	
	\medskip
	\noindent
	{\bf An efficient attack on the IBE \cite{GHPT17a}.} Our second contribution is to show that the problem is deeper than finding a new hash and sign signature scheme in rank-based cryptography to instantiate the IBE proposed in \cite{GHPT17a}. Actually the security of this IBE scheme does not solely rely on the rank metric code-based signature scheme and the rank syndrome decoding, it also relies on the Rank Support Learning (RSL) problem.
	We show here that the RSL problem is much easier for the parameters proposed in the IBE scheme \cite{GHPT17a} and can be broken by a suitable algebraic attack. 
	Interestingly enough, the approach for breaking the RSL problem is similar to what we did for RankSign:
	\begin{itemize}
		\item we exhibit a matrix code that can be deduced from the public data that contains many low-weight codewords and whose support reveals 
		the secret support of the RSL problem;
		\item we find such low weight codewords efficiently by solving a largely overdetermined bilinear system. 
	\end{itemize}
	However in this case, contrarily to the RankSign scheme, even if the set of parameters that could defeat our attack is small, it is non empty and our attack could 
	be thwarted by choosing the parameters appropriately and if an appropriate signature scheme were found.
	
	We have also explored whether it is possible to change in the IBE scheme of \cite{GHPT17a} the rank metric by the Hamming metric. It turns out that the problem 
	is much worse for the Hamming case. Indeed by adapting the IBE \cite{GHPT17a} to the Hamming metric, based on the remark that signatures must have a small weight, we show that even the simplest generic attack, namely the Prange algorithm \cite{P62}, breaks the IBE in the Hamming setting in polynomial time, and this irrespective of the way the parameters are chosen.

	\section{Generalities on rank metric and $\Fqm$-linear codes}
	\label{sec:notation}

	\subsection{Definitions and notation} 
	
	We provide here notation and definitions that are used throughout the paper.
	\newline
	
	\noindent{\bf Big O notation.} We will use the family of Bachmann-Landau notations, $f(n)=o(g(n))$, $f(n)=O(g(n))$, $f(n)=\Omega(g(n))$, $f(n)=\Theta(g(n))$, 
	$f(n)=\omega(g(n))$ meaning respectively that $\lim_{n \rightarrow \infty} \frac{f(n)}{g(n)}=0$, 
	$\limsup_{n \rightarrow \infty} \frac{|f(n)|}{g(n)} < \infty$, $\liminf_{n \rightarrow \infty} \frac{f(n)}{g(n)} >0$, 
	$f(n)=O(g(n))$ and $f(n)=\Omega(g(n))$, $\lim_{n \rightarrow \infty} \frac{|f(n)|}{|g(n)|}=\infty$.
	
	\noindent {\bf Vector notation.} Vectors will be written  using bold lower-case letters, e.g. $\xv$. The ith component of $\xv$ is denoted by $x_{i}$. Vectors are in row notation. Matrices will be written as bold capital letters, e.g. $\Xm$, and the $i$-th column of a matrix $\Xm$ is denoted $\Xm_{i}$.
	The rank of a matrix $\Xv$ will be simply denoted by $|\Xv|$.
	\newline

	\noindent {\bf Field notation.} Let 
	$q$ be a power of a prime number. We will denote by $\Fq$ the finite field of cardinality
	$q$. 
	\newline
	
	\noindent
	{\bf Coding theory notation.} A linear code $\Cc$ over a finite field $\Fq$ of length $n$ and dimension $k$ is a subspace of the vector space $\Fq^n$ of dimension $k$. We say that it has parameters $[n, k]$ or that it is an 
	$[n, k]$-code. A generator matrix $\Gm$ for it is a full rank $k \times n$ matrix over $\Fq$ which is such that
	$$\Cc =\{\uv\Gm:\uv \in \Fq^k\}.$$
	In other words, the rows of $\Gm$ form a basis of $\Cc$. A parity-check matrix $\Hm$ for it is a full-rank
	$(n-k)\times n$ matrix over $\Fq$ such that
	$$\Cc =\{\cv \in \Fq^n: \Hm \transp{\cv} =0\}.$$
	In other words, $\Cc$ is the null space of $\Hm$.

	Rank metric codes basically consist in viewing codewords as matrices. More precisely, when $N$ is the product of two numbers $m$ and $n$, $N=mn$ we will equip the vector space $\Fq^N$ with the rank metric by viewing its elements as
	matrices over $\Fq^{m \times n}$, i.e.
	$$
	d(\Xv,\Yv) = |\Xv - \Yv|.
	$$
	An $[m\times n,K]$ {\em matrix code}  of dimension $K$ over $\Fq^{m \times n}$ is a subspace of $\Fq^{m \times n}$ of dimension $K$. Such a code is equipped in a natural 
	way
	with the rank metric. There is a particular subclass of matrix codes that has the nice property to be specified much more compactly than a generic 
	matrix code. It  consists in taking 
	a linear code over an extension field $\Fqm$ of $\Fq$ of length $n$. Such a code can be viewed as a matrix code consisting of matrices 
	in $\Fq^{m \times n}$  by expressing each coordinate $c_i$ of a codeword $\cv=(c_i)_{1 \leq i \leq n}$ in a fixed $\Fq$ basis of $\Fqm$.
	When the $\Fqm$-linear code is of dimension $k$ the dimension of the matrix code viewed as an $\Fq$-subspace of $\Fq^{m \times n}$ 
	is $K=k.m$. More precisely we bring in the following definition. 
	
	\begin{definition}[Matrix code associated to an $\Fqm$ linear code] \label{def:matCode} 
		Let $\Cc$ be an $[n,k]$-linear code over $\Fqm$, that is a subspace of $\Fqm^n$ of dimension $k$ over $\Fqm$,
		and let $(\beta_1\dots \beta_m)$ be a basis of $\Fqm$ over $\Fq$. Each word $\cv \in \Cc$ can be represented by an $m\times n$ matrix $\Mat{\cv} = (M_{ij})_{\substack{1 \leq i \leq m\\ 1 \leq j \leq n}}$ over $\Fq$, with $c_j = \sum_{i=1}^m M_{ij} \beta_i$.
		The set $\{\Mat{\cv},\cv \in \Cc\}$ is  the $[m\times n,k.m]$ matrix code over $\Fq$ associated to the $\Fqm$ linear code $\Cc$. The (rank) weight of $\cv$ is defined as the rank of the associated matrix, that is $|\cv| \eqdef |\Mat{\cv}|$.
	\end{definition}

	This definition depends of course on the basis chosen for $\Fqm$. However changing the basis does not change 
	the distance between codewords. The point of defining matrix codes in this way is that they have a more compact 
	description. It is readily seen that an $[m\times n,k.m]$ matrix code over $\Fq$  can be specified from a systematic generator matrix 
	(i.e. a matrix of the form $\begin{bmatrix} \un_{k.m} | \Pm \end{bmatrix}$ with $\un_{k.m}$ being the identity matrix of size $k.m$) by  
	$k(n-k)m^2 \log_2 q$ bits
	whereas an  $\Fqm$-linear code uses only $k(n-k)\log_2 q^m=k(n-k)m \log_2 q$ bits. This is particularly interesting for cryptographic applications where this notion is directly related to the
	public key size. This is basically what explains why in general McEliece cryptosystems based on rank metric matrix codes have a smaller keysize than 
	McEliece cryptosystems based on the Hamming metric. All of these proposals (see for instance \cite{GPT91,GO01,G08,GMRZ13,GRSZ14a,ABDGHRTZ17a,AABBBDGZ17}) are actually built from matrix codes over $\Fq$ obtained 
	from 
	$\Fqm$-linear codes. In a sense, they can be viewed as structured matrix codes, much in the same way as quasi-cyclic linear codes can be viewed as structured versions of linear codes. In the latter case, the code is globally invariant by a linear isometric transform on the codewords corresponding to shifts of a certain length. In the $\Fqm$ 
	linear case the code is globally invariant by an isometric  linear transformation that corresponds to multiplication in $\Fqm$.

	\subsection{ Rank code-based cryptography} \label{subsec:rankBasedCrypto} 
	Rank-based cryptography relies on the hardness of decoding for the rank metric.
	This problem is the rank metric analogue of the well known decoding problem in the Hamming metric \cite{BMT78}.
	We give it here its syndrome formulation:

	\begin{problem}[Rank (Metric) Syndrome Decoding Problem]\label{def:RSD}$ $\\
		{\em \textup{\em Instance:}} A full-rank $(n-k) \times n$ matrix $\Hm$ over $\Fqm$ with $k \leq n$, a syndrome $\sv \in \Fqm^{n-k}$ and $w$ an integer.\\
		{\em 	\textup{\em Output:}} An error $\ev \in \Fqm^{n}$ such that $|\ev| = w$ and $\Hm\transp{\ev} = \transpose{\sv}$.

	\end{problem}

	\noindent This problem has recently been proven hard in \cite{GZ14} by a  probabilistic reduction to the decoding problem in the Hamming metric which is known to be 
	NP-complete \cite{BMT78}. This problem has typically a unique solution when $w$ is below the Varshamov-Gilbert distance $\wrGV(q,m,n,k)$ for the rank metric  which is 
	defined as
	\begin{definition}[Varshamov-Gilbert distance for the rank metric] \label{def:GV} 
		The Varshamov-Gilbert distance $\wrVG(q,m,n,k)$ for $\Fqm$ linear codes of dimension $k$ in the rank metric is defined as the smallest $t$ for which 
		$
		q^{m(n-k)} \leq B_t
		$
		where $B_t$ is the size of the ball of radius $t$ in the rank metric.
	\end{definition}
	
	\begin{remark}$ $

		\begin{enumerate}
			\item $q^{m(n-k)}$ can be viewed as the number of different syndromes $\sv \in \fqm^{n-k}$.
			\item $B_t = \sum_{i=0}^t S_i$ where $S_i$ is the size of a sphere of radius $i$ in the rank metric over $\F_q^{m \times n}$.
			This latter quantity is equal to
			$$
			S_i = \prod_{j=0}^{i-1} \frac{(q^n-q^j)(q^m-q^j)}{(q^i-q^j)} = \Theta\left(q^{i(m+n-i)} \right).
			$$
			\item 
			From this last asymptotic expression it is straightforward to check that (for more details see \cite{L14a}) 
			
			\begin{equation}\label{eq:VGdist} 
			\wrGV(q,m,n,k) = \frac{m+n - \sqrt{(m-n)^2+4 km}}{2}(1+o(1)),
			\end{equation} 
			when either $m$ or $n$ tends to infinity.
		\end{enumerate}
	\end{remark}

	\noindent The best algorithms for solving the decoding problem in the rank metric are exponential in $n^2$ as long as $m=\Theta(n)$, $w=\Theta(n)$ but $w$ stays below the Singleton bound which is defined by
	\begin{definition}[Singleton distance in the rank metric]
		The rank Singleton distance $\ws(q,m,n,k)$ for $\Fqm$ linear codes of dimension $k$  is defined as
		$
		\ws(q,m,n,k) \eqdef \left\lfloor \frac{(n-k)m}{\max(m,n)}\right\rfloor + 1
		$
	\end{definition}

	The usual notion of the support of a vector is generally relevant to decoding in the Hamming metric and corresponds for a vector $\xv=(x_i)_{1 \leq i \leq n}$ to the set of  positions
	$i$ in $\{1,\dots,n\}$ such that $x_i \neq 0$. Various decoding algorithms for the Hamming metric \cite{P62,LB88,S88,D91,FS09,BLP11,MMT11,BJMM12,MO15,DT17,BM17} use this notion 
	in a rather fundamental way. The definition of the support of a vector has to be changed a little bit to be relevant to the rank metric. This notion was first put forward in \cite{GRS13,GRS16} to obtain an analogue of the Prange decoder \cite{P62} for the rank metric.
	\begin{definition}[Support]
		Let $\xv = (x_{i})_{1 \leq i \leq n}$ be a vector of $\Fqm^{n}$, its support is defined as:
		$$
		\Sp(\xv) \eqdef \langle x_{1},\cdots,x_{n} \rangle_{\Fq}. 
		$$
	\end{definition} 
	This notion of support is among other things relied to the rank metric as it is easily verified that for any vector $\xv$ of $\Fqm^{n}$ we have:
	$$
	|\xv| = \dim(\Sp(\xv)) 
	$$

	\section{The RankSign scheme} 
	\label{sec:rkSgn}

	We recall in this section basic facts about RankSign \cite{GRSZ14a}. It is based on augmented LRPC codes. Roughly speaking it is a hash and sign signature scheme: the message $\mv$ that has to be signed is hashed by a hash function ${\Hc}$ and the signature  is equal to $f^{-1}(\Hc(\mv))$ where $f$ is a trapdoor one-way function. In this way the pair $(\mv,f^{-1}(\cH(\mv)))$ forms a valid signature. Recall now that code-based cryptography relies on Problem \ref{def:RSD} (rank syndrome decoding) which amounts to consider here the following one way-function to build a signature primitive: 	
	\begin{displaymath}
	\begin{array}{lccc}
	f_{\Hm} :  &S_{w} & \longrightarrow & \Fqm^{n-k}\\
	&\ev & \longmapsto & \ev\transpose{\Hm}
	\end{array}
	\end{displaymath}
	where $S_{w}$ denotes the words of $\Fqm^{n}$ of rank weight $w$, $\Hm$ a parity-check matrix of size $(n-k)\times n$. To introduce a trapdoor in $f_{\Hm}$ authors of \cite{GMRZ13} proposed to use  parity-check matrices of the family of augmented LRPC 
	codes. Indeed,   	
	when the underlying LRPC structure is known (roughly speaking, this is the trapdoor), there is a decoding algorithm based on the LRPC structure that computes for any (or for a good fraction)
	$\sv \in \Fqm^{n-k}$ an $\ev \in \Fqm^n$ of weight $w$ such that $\Hm \transp{\ev} = \transp{\sv}$.
	This decoding algorithm is probabilistic and the parameters of the code have to be chosen in a very specific fashion in order to have a probability of success very close to 1
	(see Fact \ref{fa:LRPCdecoder} at the end of this section).

	\noindent The following definition will be useful for our discussion. 
	\begin{definition}[Homogeneous Matrix]\label{def:homogeneous matrix}
		A matrix $\Hm=(H_{ij})_{\substack{1 \leq i \leq n-k \\ 1 \leq j \leq n}}$ over $\Fqm$ is homogeneous of weight $d$ if all its coefficients generate an $\Fq$-vector space of dimension $d$:
		\[ \dim  \left( \vsg{H_{ij}: 1 \leq i\leq n-k,\;1\leq j \leq n} \right) = d \]
	\end{definition}

	\noindent LRPC (Low Rank Parity Check) codes of weight $d$ and augmented LRPC codes of type $(d,t)$ are defined from homogeneous matrices of weight $d$ as 
	\begin{definition}[LRPC and augmented LRPC code]
		An \textup{LRPC} code over $\Fqm$ of weight $d$ is a code that admits a parity-check matrix $\Hm$ with entries in 
		$\Fqm$ that is homogeneous of 
		weight $d$ whereas an augmented \textup{LRPC} code of type $(d,t)$ over $\Fqm$ is a code that admits a parity-check matrix 
		$\Hm' = \begin{bmatrix} \Hm | \Rm \end{bmatrix} \Pm$ where $\Hm$ is a homogeneous matrix of rank $d$ over 
		$\Fqm$,  
		$\Rm$ is  a matrix with $t$ columns that has its entries in $\Fqm$ and $\Pm$ is a square and invertible matrix
		with entries in $\Fq$ that has the same number of columns as $\Hm'$. 
	\end{definition}
	
	\begin{remark} Note that any invertible
		$\Pm \in \Fq^{n \times n}$ is an isometry for the rank metric, since
		for any $\xv \in \Fqm^n$ we have $\Sup(\xv) = \Sup(\xv \Pm)$ and therefore
		$$
		|\xv| = |\xv \Pm|.
		$$ 
	\end{remark}

	\noindent The public key and the secret key for RankSign are given by:\\
	\newline 
	{\bf public key:} $\Hpub$ which is a random $(n-k)\times n$ parity-check matrix of an augmented LRPC code of type $(d,t)$. It is of the 
	form
	$$
	\Hpub = \Qm \Hm'
	$$
	with 
	$
	\Hm' = \begin{bmatrix} \Hm | \Rm \end{bmatrix} \Pm
	$
	where $\Qm$ is an invertible $(n-k)\times(n-k)$ matrix over $\Fqm$, 
	$\Hm$ is a homogeneous matrix of rank $d$ over 
	$\Fqm$,  
	$\Rm$ is  a matrix with $t$ columns that has its entries in $\Fqm$ and $\Pm$ is a square and invertible matrix
	with entries in $\Fq$ that has the same number of columns as $\Hm'$.\\
	\newline
	\noindent
	{\bf secret key:} The matrix $\Hsec \eqdef \begin{bmatrix} \Hm|\Rm \end{bmatrix}$.
	\newline

	\noindent From the knowledge of this last matrix a signature is computed by using a decoding algorithm devised for LRPC codes.
	Recall that LRPC codes can be viewed as analogues of LDPC codes for the rank metric. In particular, they enjoy an efficient decoding algorithm based on their low rank parity-check matrix. Roughly speaking, Algorithm 1 of \cite{GMRZ13} decodes up to $w$ errors when $dw \leq n-k$ in polynomial time (see \cite[Theorem 1]{GMRZ13}).
	It uses in a crucial way the notion of the linear span of a product of subspaces of $\Fqm$:
	\begin{definition} 
		Let $U$ and $V$ be two subspaces of $\Fqm$, then 
		$$
		U\cdot V \eqdef \langle uv : u \in U,\;v \in V\rangle_{\Fq}.
		$$	
	\end{definition}  
	\noindent Roughly speaking, Algorithm 1 of \cite{GMRZ13} works as follows when we have to recover an error $\ev$ of weight $w$ from the knowledge of its syndrome $\sv$ with respect to a 
	parity-check matrix $\Hm=(H_{ij})_{\substack{1 \leq i \leq n-k \\ 1 \leq j \leq n}}$ over $\Fqm$ that is homogeneous of weight $d$, that is 
	\begin{equation}
	\label{eq:fundamental}
	\transp{\sv} = \Hm \transp{\ev}.
	\end{equation}
	\begin{enumerate}
		\item Let $U \eqdef \vsg{H_{ij}: 1 \leq i\leq n-k,\;1\leq j \leq n}$, $V \eqdef \Sup(\ev)$  and  $W \defeq \Sup(\sv)$.
		$U$ and $W$ are known, whereas $V$ is unknown to the decoder. By definition $U$ is of dimension $d$ and it is convenient to bring in 
		a basis $\{f_1,\dots,f_d\}$ for it.
		\item It turns out that we typically have 
		$W = U\cdot V$. Moreover it is clear that in such a case
		$V \subset f_1^{-1}W \cap f_2^{-1}W \cdots f_d^{-1}W$. It also turns out that we typically have
		$$
		V = f_1^{-1}W \cap f_2^{-1}W \cdots f_d^{-1}W.
		$$
		$V$ is therefore computed by taking the intersection of all the $f_i^{-1}W$'s.
		\item Once we have the support of $\ev$ ($V=\Sup(\ev)$), the error $\ev=(e_1,\dots,e_n)$ can be recovered by solving the linear equation 
		$\Hm\transpose{\ev}=\transpose{\sv}$ with the additional constraints $e_i \in \Sup(\ev)$ for 
		$i \in \{1,\ldots,n\}$. There are in this case enough linear constraints to recover a unique $\ev$.
	\end{enumerate}

	\noindent The last algorithm seems to apply when there is a unique solution to \eqref{eq:fundamental}. It can also be used with a slight modification (by 
	adding ``erasures'' \cite{GRSZ14a}) for weights for which there are many solutions to it (this is typically the regime which is used for the RankSign scheme).
	It namely turns out, see \cite{GRSZ14a}, that this decoder can for a certain range of parameters be used for a large fraction of possible syndromes $\sv \in \Fqm^{n-k}$ to produce 
	an error $\ev$ of weight $w$ that satisfies \eqref{eq:fundamental}. It can even be required that $\Sup(\ev)$ contains a subspace $T$ of some dimension $t$. 
	Furthermore this procedure can also be generalized to a parity-check matrix of an augmented LRPC code.  More precisely to summarize
	the discussion that can be found in \cite{GRSZ14a,AGHRZ17}
	\begin{fact}\label{fa:LRPCdecoder}
		Let $\Hm$ be a random homogeneous matrix of weight $d$ in $\Fqm^{(n-k) \times n}$, 
		$\Hm' =  \begin{bmatrix} \Hm | \Rm \end{bmatrix} \Pm$ where 
		$\Rm$ is  a matrix with $t$ columns that has its entries in $\Fqm$ and $\Pm$ is a square and invertible matrix
		with entries in $\Fq$ that has the same number of columns as $\Hm'$.
		There is a probabilistic polynomial time algorithm that outputs for a large fraction of syndromes $\sv \in \Fqm^{n-k}$, subspaces $T$ of $\Fqm$ of $\Fq$--dimension $t'$,  an error $\ev$ of weight $w$ whose support contains the subspace $T$ that satisfies
		$$
		\Hm' \transp{\ev} = \transp{\sv}
		$$
		as soon as the parameters $n,k,t,t',d,w$ satisfy
		\begin{eqnarray}
		m &= &(w-t')(d+1)\\
		n-k & = &d(w-t-t') \label{eq:w}\\
		n & = &(n-k)d. \label{eq:problem}
		\end{eqnarray}
	\end{fact}
	
	\section{Identity-Based-Encryption in code-based cryptography}
	\label{sec:IBEscheme} 
	
	We recall in this section the \cite{GHPT17a} approach for obtaining an IBE scheme whose security relies  on code-based assumptions. In some sense, this scheme 
	can be  viewed as an adaptation of the first quantum-safe IBE which was introduced by \cite{GPV08} in the paradigm of lattice-based cryptography. It relies among other things on two fundamental building blocks: a hash and sign primitive and an encryption scheme related to it. The adaptation relies on two building blocks:
	$i)$ a signature scheme, RankSign whose security relies on code-based assumptions for the rank metric, $ii)$ a new encryption scheme, namely RankPKE \cite{GHPT17a}, 
	based on the Rank Support Leaning (RSL) problem.
	\cite{GHPT17a}  gives a security proof of the IBE scheme that relies on two assumptions:
	$i)$ the key security of RankSign and $ii)$ the difficulty of RSL. Furthermore, the work of \cite{GHPT17a} can be easily generalized  to the more common Hamming metric.
	It is why we present in what follows the \cite{GHPT17a} IBE scheme with codes independently of the metric.

	Roughly speaking, an IBE is a specific public-key encryption scheme that allows senders to encrypt messages thanks to the receiver's identity (such as its email address). To permit this protocol there is a third party, say a Key Derivation Center, which owns a master secret-key MSK and an associated public-key MPK that allows to compute from \textit{any} identity $id$ a related secret quantity $\textup{sk}_{id}$ that will be used in a public-key encryption scheme involving an arbitrary sender and the receiver of identity $id$, with the pair of public/secret key $((id,\textup{MPK}),\textup{sk}_{id})$. 
	In this paradigm any identity $id$ needs to be matched with a secret key $\textup{sk}_{id}$ and to achieve this goal it was proposed in \cite{GPV08} to use a hash and sign primitive. Roughly speaking, for a trapdoor function $f$ and a hash function $\cH$ the Key Derivation Center will compute from $id$ the quantity $f^{-1}(\cH(id))$ which will be used as $\textup{sk}_{id}$. We summarize in Figure \ref{fig:IBEgene} how this IBE works.   In the case of \cite{GPV08}, signatures sample short vectors whose addition with the hash of the identity gives lattice points. Then this is used as a secret-key of an encryption scheme whose security relies on the hardness of the LWE problem (see \cite[Section 7.1, p26]{GPV08}).

	\begin{figure}
		\begin{center} 
			\begin{tikzpicture}
			\draw[fill = purple!7] (0,0) rectangle (6,1.8);
			\draw[fill = blue!7] (0.4,0.2) rectangle (2.2,1);
			\draw[fill = blue!7] (3.8,0.2) rectangle (5.6,1);
			\draw[->] (2.3,0.55) -- (3.7,0.55);
			\node at (3,1.5) {Key Derivation Center};
			\node at (1.3,0.75) {\small Key};
			\node at (1.3,0.38) {\small Setup};
			\node at (4.7,0.75) {\small Key};
			\node at (4.7,0.38) {\small Derivation};
			\node at (3,0.85) {\small MSK};
			\node at (3,0.25) {\small $id$};
			\draw[fill = purple!7] (4.7,-3.7) rectangle (7.5,-2);
			\draw[fill = blue!7] (4.8,-3) rectangle (6.8,-2.2);
			\draw[->] (4.7,-0.15) -- (5.8,-2.15);
			\node at (7.2,-1) {$\textup{sk}_{id}\eqdef f^{-1}(\cH(id))$};
			\node at (6.5,-3.4) {Receiver $id$};
			\node at (5.8,-2.6) {\small Decryption};
			\draw[->] (6.9,-2.6) -- (9,-2.6);
			\node at (8.5,-2.3) {$\small \textup{Dec}_{\textup{sk}_{id}}(\cv)$}; 
			\node at (-0.3,-1) {$\textup{MPK}$};
			\draw[fill = purple!7] (-1.5,-3.7) rectangle (1.3,-2);
			\draw[fill = blue!7] (1.2,-3) rectangle (-0.8,-2.2);
			\draw[->] (1.3,-0.15) -- (0.2,-2.15);
			\draw [->] (1.5,-2.6) -- (4.5,-2.6);
			\node at (3,-2.25) {\small $\cv \eqdef \textup{Enc}_{(\textup{MPK},id)}(\mv)$};
			\node at (0.2,-2.6) {\small Encryption};
			\draw[->] (-3,-2.6) -- (-1,-2.6);
			\node at (-2.3,-2.3) {$(\mv,id)$}; 
			\node at (-0.8,-3.4) { Sender};
			\end{tikzpicture}	
		\end{center} 
		\caption{IBE in the GPV context}
		\label{fig:IBEgene}
	\end{figure}
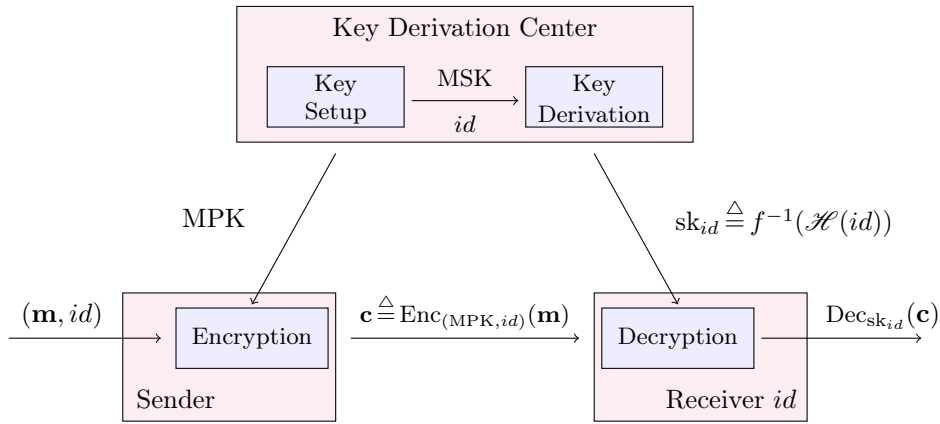 
	
	\medskip
	\noindent
	{\bf IBE in code-based cryptography.} We give now the general framework of \cite{GHPT17a} for obtaining a code-based IBE scheme. It is only given in the rank metric
	case in \cite{GHPT17a}, but the approach is really more general than this and can be given for the Hamming metric too. We will detail what happens for both metrics here. 
	As explained above, this scheme builds upon a hash and sign primitive and the authors of \cite{GHPT17a} proposed RankSign there but 
	in our description the signature scheme is just a black-box.

	Let $\cC_{\sgn}$ be a code of length $n_{\sgn}$ and dimension $k_{\sgn}$ for which there is a trapdoor that enables to compute for any $\yv \in \F_{2}^{n_{\sgn}}$ a codeword $\cv_{\yv} \in \cC_{\sgn}$ at distance $w_{\sgn}$. Let $w_{\dec} $ be an integer, $\cC_{\dec}$ be a code of length $n_{\dec}$ and dimension $k_{\dec}$ such that it exists a polynomial algorithm to decode a linear (in the length) error weight. Let $\Gm_{\cC_{\sgn}}$ and $\Gm_{\cC_{\dec}}$ be generator matrices of the codes $\cC_{\sgn}$ and $\cC_{\dec}$ respectively. Then it is proposed in \cite{GHPT17a} to set master secret and public keys as:
	\begin{itemize} 
		\item MSK be the trapdoor which enables to decode at distance $w_{\sgn}$ in $\cC_{\sgn}$; 	
		\item  $\textup{MPK} \eqdef \left( \cC_{\sgn},\cC_{\dec} \right)$. 
	\end{itemize} 
	Let $id$ be an identity and $\cH$ be a hash function whose range is $\F_{2}^{n_{\sgn}}$ or $\Fqm^{n_{\sgn}}$ according to the metric which is used. The key derivation center  computes with MSK and $id$ a vector $\uv_{id}$ such that:
	\begin{equation} \label{eq:IBEsgn} 
	|\uv_{id}\Gm_{\cC_{\sgn}} - \cH(id)| = w_{\sgn}\mbox{ where }| \cdot | \mbox{ denotes either the Hamming or rank metric}
	\end{equation}
	This is used as the secret key associated to the identity $id$:
	\begin{itemize}
		\item $\textup{sk}_{id} \eqdef \uv_{id}$. 
	\end{itemize}
	We are now ready to present the encryption scheme whose public/secret key is $((\Gm_{\cC_{\sgn}},\Gm_{\cC_{\dec}},id),\uv_{id})$ and which in the particular case of the rank metric is the RankPKE 
	scheme introduced in \cite{GHPT17a}. This primitive is related to the work of Alekhnovich \cite{A11}.

	\begin{itemize} 
		\item {\bf Encryption.} Let $\mv$ be the message that will be encrypted. We will denote by $\F$ the finite field $\F_{2}$ or $\Fqm$ depending on the Hamming or rank metric. The authors of \cite{GHPT17a} introduced the trapdoor function:
		\begin{displaymath}
		\begin{array}{lccc}
		g_{\Gm_{\cC_{\sgn}},\Gm_{\cC_{\dec}},id} :  &\F^{k_{\dec}} & \longrightarrow & \F^{(k_{\sgn}+1) \times n_{\dec}}\\
		&\mv& \longmapsto &\begin{bmatrix} \Gm_{\cC_{\sgn}}\Em \\ \cH(id)\Em + \mv\Gm_{\cC_{\dec}} \end{bmatrix} 
		\end{array}
		\end{displaymath}
		where $\Em$ has a size $n_{\sgn} \times n_{\dec}$. In the case of the rank metric $\Em$ is a matrix uniformly picked at random among the homogeneous matrices of weight $w_{\dec}$ and in the case of the Hamming metric, $\Em$ is  picked uniformly at random among the matrices whose columns have all  weight $w_{\dec}$.

		\item {\bf Decryption.} The secret key $\uv_{id}$ is  used as
		\begin{align*} 
		(\uv_{id},-1)g_{\Gm_{\cC_{\sgn}},\Gm_{\cC_{\dec}},id}(\mv) &= (\uv_{id},-1) \begin{bmatrix} \Gm_{\cC_{\sgn}}\Em \\ \cH(id)\Em + \mv\Gm_{\cC_{\dec}} \end{bmatrix} \\
		&= \left( \uv_{id}\Gm_{\cC_{\sgn}} - \cH(id)\right)\Em - \mv\Gm_{\cC_{\dec}} \\ 
		\end{align*} 
		It can be verified that under certain restrictions on $w_{\sgn}$ and $w_{\dec}$, the weight of the vector $\left( \uv_{id}\Gm_{\cC_{\sgn}} - \cH(id)\right)\Em$ is low enough, so that a decoding algorithm
		for $\cC_{\dec}$ will recover $\mv$. 
		The following proposition gives a constraint on these parameters so that decoding is possible in principle. 
		\begin{proposition} \label{propo:paramIBE} 
			In order to be able to decode asymptotically at constant rate $R$, there should exist an $\varepsilon(R) >0$ 
			such that all the parameters $n_{\sgn}, w_{\sgn}$ and $w_{\dec}$ have to verify 
			\begin{itemize}
				\item in the rank metric case
				\begin{equation} \label{eq:paramIBEH}
				w_{\sgn}w_{\dec} = (1-\varepsilon(R)) \min(m,n_{\dec})
				\end{equation} 
				\item in the Hamming metric case
				\begin{equation}
				\label{eq:paramIBER}
				w_{\sgn}w_{\dec} = O(n_{\sgn}).
				\end{equation}
			\end{itemize}
		\end{proposition}

		\begin{proof} We separate the proof in two parts.

			{\em Rank metric.} In this case, as proved in \cite[\S 3.2]{GHPT17a} the rank weight of the error term $\left( \uv_{id}\Gm_{\cC_{\sgn}} - \cH(id)\right)\Em$ is 
			with high probability $w_{\sgn}w_{\dec}$. Recall that $\cC_{\dec}$ is a code over the alphabet $\Fqm$. A necessary condition to be able to decode with a fixed rate code is that the dimension of the support of the error 
			is at most some fraction of 
			the dimension $m$ of the whole space $\Fqm$ and of the length $n_{\dec}$ of the code we decode.
			This means that $w_{\sgn}w_{\dec} \leq (1-\varepsilon(R)) \min(m,n_{\dec})$.

			{\em Hamming metric.} Recall that $\uv_{id}\Gm_{\cC_{\sgn}} - \cH(id)$ has Hamming weight $w_{\sgn}$ (see \eqref{eq:IBEsgn}). It is easily verified that the probability for one bit of $\left( \uv_{id}\Gm_{\cC_{\sgn}} - \cH(id)\right)\Em$ to be equal to $1$ is of the form $(1/2)(1-e^{-2w_{\sgn}w_{\dec}/n(1+O(1))})$ when the columns of $\Em$ are picked uniformly at random among the words of Hamming weight $w_{\dec}$. Furthermore 
			the relative weight of $\left( \uv_{id}\Gm_{\cC_{\sgn}} - \cH(id)\right)\Em$ concentrates around this probability 
			and a necessary condition to be able to decode at constant rate asymptotically is that this relative weight is a constant $<1/2$. Therefore it is necessary to have $w_{\sgn}w_{\dec} = O(n_{\sgn})$.  
			
		\end{proof} 
		
	\end{itemize}
	
	The constraint set on the parameters by this proposition is crucial to instantiate the IBE in code-based cryptography. 
	Unfortunately, this constraint implies a fatal weakness for the Hamming based scheme and a hard to meet condition for the rank metric in order to have
	a secure scheme as we will see in what follows.
	\newline

	\noindent
	{\bf The RSL problem.} We recall here the assumption upon which the security of RankPKE relies (the previous encryption scheme in rank metric), 
	namely the Rank Support Leaning (RSL) problem introduced in \cite{GHPT17a}. This problem is a rank syndrome decoding problem with syndromes 
	that are associated to errors that all share the same support which is the secret.

	\begin{problem}[\textup{RSL} - Rank Support Learning]\label{prob:RSL}$ $\\
		\textup{\em Parameters:} $n,k,N,w$\\
		\textup{\em Instance:} $(\Am,\Am\Em)$ where $\Am$ is a full rank matrix of size $(n-k)\times n$, $\Em$ a matrix of size 
		$n \times N$ where all its coefficients belong to a same subspace $F$ of $\Fqm$ of dimension $w$ \\
		\textup{\em Output:} the subspace $F$.\\
		The decisional version of \textup{RSL}, namely \textup{DRSL}, is to distinguish distributions $(\Am,\Am\Em)$ from $(\Am,\Rm)$ where $\Am,\Rm$ and $\Em$ are random variables whose distribution is uniform over matrices of size $(n-k)\times n, (n-k) \times N$ and over homogeneous matrices of size $n \times N$ and weight $w$.
	\end{problem} 
	
	\begin{remark} Let $(\Am,\Am\Em)$ be an instance of \textup{RSL}. The matrix $\Am$ is of full-rank of size $(n-k)\times n$ and we can perform Gaussian elimination on its rows to get a matrix $\Sm$ such that $\Sm\Am = \lbrack I_{n-k}|\Am'\rbrack$. The pair $(\Sm\Am,\Sm\Am\Em)$ is still an instance of \textup{RSL} with the same parameters and secret subspace $F$, it is why we can always assume that for any instance of \textup{RSL} the matrix $\Am$ is in systematic form.

	\end{remark}

	As proved in \cite[\S 3.3, p13, Theorem 1]{GHPT17a} the security of RankPKE relies on the DRSL problem.

	\section{Attack on RankSign}

	\subsection{The problem with RankSign : low rank codewords in the augmented LRPC code} 
	\label{sub:lowdistrib}
	
	A natural way to attack RankSign is to find low weight codewords in the dual of the augmented LRPC code. Recall that the 
	public parity-check matrix used in the scheme is a matrix $\Hpub$ where
	$$
	\Hpub = \Qm \Hm'
	$$
	with 
	$\Hm' = \begin{bmatrix} \Hm | \Rm \end{bmatrix} \Pm$ where $\Hm$ is a homogeneous matrix of rank $d$ over 
	$\Fqm$,  
	$\Rm$ is  a matrix with $t$ columns that has its entries in $\Fqm$, $\Pm$ is a square and invertible matrix
	with entries in $\Fq$ that has the same number of columns as $\Hm'$ and $\Qm$ is a square and invertible matrix over $\Fqm$ which has the same number of rows as $\Hm'$.
	If we call $\Cpub$ the ``public code'' with parity-check matrix 	$\Hpub$, then 
	the dual code $\Cpub^\perp$ that has for generator matrix $\Hpub$ has codewords of weight $\leq d+t$ since 
	rows of $\Hm' \Pm$ belong to this code, and all of its rows have rank weight $\leq d+t$ since the rows of 
	$\Hm'$ have weight at most $d+t$ and $\Pm$ is an isometry for the rank metric.
	The authors have chosen the parameters of the RankSign scheme so that finding codewords of weight 
	$t+d$ in $\Cpub^\perp$ is above the security level of the scheme. However, it turns out that 
	due to the peculiar parameters chosen in the RankSign scheme (see Fact \ref{fa:LRPCdecoder}), $\Cpub$ has many very low weight codewords.
	This is the main problem in RankSign. Before we give a precise statement together with its proof, we
	will give a general result showing that LRPC codes may have under certain circumstances low weight codewords.

	\begin{lemma} \label{lem:pcp}
		Let $\cC$ be an \textup{LRPC} code of length $n$ and dimension $k$ over $\Fqm$ that is associated to an homogeneous matrix $\Hm$ that has all its entries in 
		a subspace $F$ of $\Fqm$. Furthermore we suppose there exists a subspace $F'$ of $\Fqm$ such that 
		$$
		(n-k) \dim(F\cdot F') < n \dim F'.
		$$
		Then there exist non-zero codewords in the \textup{LRPC} code whose support is included in $F'$. They are therefore of rank weight at most $\dim F'$.
		Furthermore this set of codewords, that is 
		$$\cC' \eqdef \left\{\cv \in \cC: c_i \in F',\;\forall i \in \IInt{1}{n} \right\}$$ 
		forms an $\Fq$ subspace of $\Fqm^n$ that is of dimension $\geq n \dim F' -(n-k) \dim(F\cdot F')$.
	\end{lemma}

	\begin{proof} Denote the entry in row $i$ and column $j$ of $\Hm$ by $H_{i,j}$. A codeword $\cv$ of the LRPC code satisfies
		\begin{equation} \label{eq:vec} 
		\forall i \in \IInt{1}{n-k},  \quad \sum_{j=1}^{n} H_{i,j}c_{j} = 0.
		\end{equation} 
		Looking in addition for a codeword $\cv$ that has all its entries in $F'$ and expressing these $n-k$ linear equations 
		over $\Fqm$ in a basis of $F \cdot F'$ (since 
		$\sum_{j=1}^{n} H_{i,j}c_{j}$ belongs by definition to $F \cdot F'$) and expressing each $c_j$ in a $\Fq$ basis $\{f'_1,\dots,f'_{d'}\}$ of $F'$ as
		$c_j = \sum_{\ell=1}^{d'} c_{j,\ell}f'_\ell$ we obtain $(n-k) \dim(F \cdot F')$ linear equations over $\Fq$ involving $n \dim F'$ unknowns (the $c_{j,\ell}$'s) in $\Fq$.
		The solution space is therefore of dimension greater $\geq n \dim F' -(n-k) \dim(F\cdot F')$.  
	\end{proof}

	\begin{remark} This theorem proves the existence of low rank codewords in an \textup{LRPC}-code under some conditions but it does not give any efficient way to find them.
	\end{remark}

	By using this lemma, we  will prove the following corollary that explains that the augmented LRPC codes that are
	used in the RankSign signature necessarily contain many rank weight $2$ codewords. This is in a sense a consequence of the constraint \eqref{eq:problem} on the parameters of RankSign.	
	\begin{corollary} \label{cor:rksgn} 
		Let $\Cpub$ be an $\lbrack n+t,k+t \rbrack$ public code of \textup{RankSign} over $\Fqm$ which has been obtained from an $\lbrack n,k \rbrack$ \textup{LRPC}-code that is associated to a homogeneous matrix $\Hm$ that has all its entries in 
		an $\Fq$ subspace $F$ of $\Fqm$. Consider a subspace $F'$ of $F$ of dimension $2$ and let
		$$
		\Cpub' \eqdef \left\{\cv \in \Cpub: c_i \in F',\;\forall i \in \llbracket 1,n+t\rrbracket \right\}.
		$$
		$\Cpub'$ is an $\Fq$ subspace of $\Fqm^{n+t}$.
		If \eqref{eq:problem} holds, that is
		$n = (n-k) d$, then 
		$$
		\dim_{\Fq} \Cpub' \geq n/d.
		$$
	\end{corollary}

	\begin{proof} Let 
		$\Hpub\in \Fqm^{(n-k)\times (n+t)}$
		be the public parity-check matrix for the RankSign public code $\Cpub$. Recall that $\Hpub$ has been obtained as $\Hpub = \Qm \begin{bmatrix}  \Hm | \Rm \end{bmatrix} \Pm$ where: 
		\begin{itemize}
			\item $\Pm$ is a non-singular matrix with entries in  $\Fq$ of size $(n+t) \times (n+t)$,

			\item $\Qm$ is an invertible matrix of $\Fqm$ of size $(n-k) \times (n-k)$,

			\item $\Rm$ is a random matrix of $\Fqm$ of size $(n-k) \times t$,

			\item $\Hm$ is a homogeneous  $(n-k) \times n$ matrix of weight $d$ with all its entries in $F$. 		\end{itemize}
		Choose a basis $\{x_1,x_2,\dots,x_d\}$ of $F$ such that $\{x_1,x_2\}$ is a basis of $F'$.
		We observe now that
		$$
		F \cdot F' = \vsg{x_i x_j: i \in \IInt{1}{d},\;j \in \IInt{1}{2}}.
		$$
		The cardinality of the set $\{x_i x_j:  i \in \IInt{1}{d},\;j \in \IInt{1}{2}\}$ is actually $2d-1$ 
		because $x_1 x_2 = x_2 x_1$.
		This implies that 
		$$
		\dim(F \cdot F') \leq 2d - 1.
		$$
		It leads to the following inequalities,
		\begin{align*}
		n \dim(F') - (n-k) \dim(F\cdot F') &\geq 2n - (n-k)(2d-1) \\ 
		&=  2d(n-k) - (n-k)(2d-1) \mbox{ (since $n=(n-k)d$)} \\
		&= n-k \\
		&= \frac{n}{d} \quad\mbox{  (since $n=(n-k)d$)}. 
		\end{align*}
		Let $\Clrpc$ be the LRPC code of weight $d$ associated to the parity-check matrix $\Hm$ and let $\Clrpc'$ be 
		an $\Fq$ subspace of it  that is defined by
		$$
		\Clrpc' \eqdef \left\{\cv \in \Clrpc: c_i \in F',\;\forall i \in \llbracket 1,n\rrbracket \right\}.
		$$ 
		By applying Lemma \ref{lem:pcp} we know that
		\begin{equation}\label{eq:clrpc'}
		\dim_{\Fq} \Clrpc' \geq \frac{n}{d}.
		\end{equation}
		Consider now
		$$
		\Cpub' \eqdef \{(\clrpc,\mathbf{0}_{t})\transp{(\Pm^{-1})}: \clrpc \in \Clrpc'\},
		$$
		where  $\mathbf{0}_{t}$  denotes the vector with $t$ zeros.
		From \eqref{eq:clrpc'} we deduce that
		$$
		\dim_{\Fq}\Cpub'\geq \frac{n}{d}.
		$$
		Moreover the entries of any element $\cv'$ in $\Cpub'$ belong to $F'$ because the entries of $\Pm$ are in 
		$\Fq$. Let us now prove  that $\Cpub'$ is contained in $\Cpub$.  To verify this, consider an element
		$\cv'$ in $\Cpub'$. It can be written as
		$$
		\cv'= (\clrpc,\mathbf{0}_{t})\transp{(\Pm^{-1})}. 	$$
		We observe now that
		\begin{align*} 
		\Hpub \transp{\cv' }&= \Hpub \Pm^{-1}\transp{(\clrpc,\mathbf{0}_{t})}\\
		&= \Qm \begin{bmatrix} \Hm | \Rm \end{bmatrix} \Pm \Pm^{-1}\transp{(\clrpc,\mathbf{0}_{t})}\\
		&= \Qm\begin{bmatrix} \Hm | \Rm \end{bmatrix}\transp{(\clrpc,\mathbf{0}_{t})} \\
		&= \Qm \Hm \transp{\clrpc} \quad ( \;\Rm \in \Fqm^{(n-k) \times t} \;) \\
		&= \mathbf{0} \quad  (\clrpc \mbox{ belongs to the code of parity-check matrix } \Hm)
		\end{align*} 
		This proves that $\Cpub' \subset \Cpub$ which concludes the proof.  
	\end{proof}

	\subsection{Weight $1$ codewords in a projected code}\label{ss:w1pc}
	Corollary \ref{cor:rksgn} shows that there are many weight $2$ codewords in $\Cpub$. We can even restrict our 
	search further by noticing that without loss of generality we may assume that the space $F$ in which 
	the entries of the secret parity-check matrix $\Hm$ of the LRPC code are taken contains $1$. Indeed,  for any $\alpha$ in $\Fqm^{\times}$, $\alpha \Hm$ is also a parity-check matrix of the LRPC code
	and has its entries in $\alpha F$. By choosing $\alpha$ such that $\alpha F$ contains $1$ we get our claim.
	
	Consider now a  supplementary space $V$ of $\vsg{1}=\Fq$ with respect to $\Fqm$, that is an $\Fq$-space
	of dimension $m-1$ such that 
	$$
	\Fqm = V \oplus \Fq.
	$$
	The previous discussion implies that there is a matrix-code in $\Fq^{(m-1)\times (n+t)}$, deduced from $\Cpub$ by projecting the entries onto $V$, that contains codewords of weight $1$.
	More specifically,  consider an $\Fq$ basis $\{\beta_1,\beta_2, \cdots, \beta_{m}\}$ of $\Fqm$ such that 
	$\beta_{m}=1$ and for 
	$\cv=(c_i)_{1 \leq i \leq n+t} \in \Fqm^{n+t}$ consider 
	$$
	\Matp{\cv} = (M_{ij})_{\substack{1 \leq i \leq m-1\\ 1 \leq j \leq n+t}} \in \Fq^{(m-1)\times (n+t)}  
	$$
	where $c_j = \sum_{i=1}^m M_{ij} \beta_i$.
	Now let
	$\Cpubp$ be the matrix-code in $\Fq^{(m-1)\times (n+t)}$ defined by
	$$
	\Cpubp \eqdef \left\{ \Matp{\cv} : \cv \in \Cpub \right\}.
	$$ 
	It is clear that
	\begin{fact}\label{fac:rw1}
		$\Cpubp$ contains codewords of rank weight $1$.
	\end{fact}

	\noindent These are just the codewords $\cv'$ which are of the form $\Matp{\cv}$ where $\cv \in \Cpub'$ with 
	$\Cpub'$ being defined from a subspace $F'$ of $F$ that contains $1$ (we can make this assumption since we can assume that $F$ contains 
	$1$).

	$\Cpubp$ has the structure of an $\Fq$-subspace of $\Fq^{(m-1)\times(n+t)}$. It is typically of dimension $(k+t)m$ 
	(i.e. the same as the $\Fq$ dimension of $\Cpub$). 
	Moreover once we have these rank weight $1$ codewords in $\Cpubp$ we can lift them to obtain rank weight $\leq 2$ 
	codewords in $\Cpub$ because for any $\cv \in \Cpub$ the last row of $\Mat{\cv}$ can be uniquely recovered from
	$\Matp{\cv}$ by performing linear combinations of the entries of $\Matp{\cv}$. We call this operation deducing $\cv$ from $\Matp{\cv}$ {\em lifting}
	from $\Cpubp$ to $\Cpub$.
	
	\subsection{Outline of the attack} \label{subsec:outlineAtt}
	Finding codewords of rank $1$ in $\Cpubp$ obviously reveals much of the secret LRPC structure. Lifting elements 
	in $\Cpubp$ that are of rank $1$ to $\Cpub$ as explained at the end of Subsection \ref{ss:w1pc} yields codewords of $\Cpub$ that have typically rank weight $2$. This can be used to reveal $F'$ and actually the whole subspace $F$ by finding enough rank $1$ codewords in $\Cpubp$. Once $F$ is recovered a suitable form for a parity-check matrix 
	of $\Cpub$ can be found that allows signing like a legitimate user. For the case of the parameters of RankSign 
	proposed 
	in \cite{GRSZ14a,AGHRZ17} 
	for which we always have $d=2$ we will proceed slightly differently here.
	Roughly speaking, our attack can be decomposed as follows
	\begin{enumerate}
		\item We find a particular element $\Mm$ in $\Cpubp$ of rank weight $1$  by solving a certain bilinear system with Gr\"obner bases techniques.
		\item\label{step:two} We lift $\Mm \in \Cpubp$ to $\cv \in \Cpub$ and compute $F' \eqdef \Sup(\cv)$.
		\item\label{step:three} We compute from $F'$ the $\Fq$-subspace \\ $\Cpub' \eqdef \left\{\cv=(c_i)_{1 \leq i \leq n+t} \in \Cpub: c_i \in F'\;\forall i \in \llbracket 1,n+t\rrbracket \right\}$. When $d=2$ this set has typically dimension $k$.
		\item We use this subspace  of $\Cpub$ to find a suitable parity-check matrix for $\Cpub$ which allows us 
		to sign like a legitimate user.
	\end{enumerate}
	Steps 2. and 3. are straightforward. We just give details for Steps 1. and 4. in what follows.
	
	\subsection{Finding rank $1$ matrices in $\Cpubp$ by solving a bilinear system}
	
	\noindent	
	{\bf The basic bilinear system.}
	Finding rank $1$ matrices in $\Cpubp$ can be formulated as an instance of the MinRank problem
	\cite{BFS99,C01}. We could use standard techniques for solving this problem \cite{KS99,FLP08,FSS10,S12} but 
	we found that it is better here to use the algebraic modelling suggested in \cite{AGHRZ17}. It basically consists in setting up an algebraic system with unknowns $\xv=(x_1,\dots,x_{m-1}) \in \Fq^{m-1}$ and $\yv \in \Fq^{n+t}$
	where the unknown matrix $\Mm$ in $\Cpubp$ that should be of rank $1$ has the form
	$$
	\Mm = \begin{pmatrix} x_1 y_1 & x_1 y_2 & \hdots & x_1 y_{n+t} \\
	x_2 y_1 & x_2 y_2 & \hdots & x_2 y_{n+t} \\
	\vdots & \vdots & \vdots  & \vdots \\
	x_{m-1} y_1 & x_{m-1} y_2 & \hdots & x_{m-1} y_{n+t}
	\end{pmatrix}.
	$$
	Recall that $\Cpubp$ has the structure of an $\Fq$ subspace of $\Fq^{(m-1)\times(n+t)}$ of dimension
	$(k+t)m$. By viewing the elements of $\Cpubp$ as vectors of $\Fq^{(m-1)(n+t)}$, i.e. the matrix $\Mm=(M_{ij})_{\substack{1 \leq i \leq m-1\\ 1 \leq j \leq n+t}}$ is viewed as the vector $\mv=(m_\ell)_{1 \leq \ell \leq (m-1)(n+t)}$ where
	$m_{(i-1)(n+t)+j} = M_{i,j}$, we can compute a parity-check matrix $\Hpubp$ for it. It is an 
	$((m-1)(n+t)-(k+t)m)\times (m-1)(n+t)$ matrix that we denote by $\Hpubp = (H^{\textup{proj}}_{ij})_{\substack{1 \leq i \leq (m-1)(n+t)-(k+t)m \\ 1 \leq j \leq (m-1)(n+t) }}$. This matrix gives $(m-1)(n+t)-(k+t)m$ bilinear equations that have
	to be satisfied by the $x_i$'s and the $y_j$'s:
	
	\begin{equation} \label{syst:eqRkSgn} 
	\left\{
	\begin{array}{l}
	\mathop{\sum}\limits_{j = 1}^{n+t} \mathop{\sum}\limits_{i=1}^{m-1} H^{\textup{proj}}_{ 1,(i-1)(n+t) + j} x_i y_{j} = 0 \\
	\qquad\vdots \\ 
	\mathop{\sum}\limits_{j = 1}^{n+t}\mathop{\sum}\limits_{i=1}^{m-1} H^{\textup{proj}}_{ (n+t)(m-1) - (k+t)m,(i-1)(n+t) + j} x_{i}y_{j} = 0 \\
	\end{array}
	\right.
	\end{equation}  
	
	\noindent {\bf Restricting the number of solutions.} 
	We have solved the bilinear system \eqref{syst:eqRkSgn} with standard Gr\"obner bases techniques that are implemented in Magma.
	To speed-up the resolution of the bilinear system with Gr\"obner bases techniques 
	(especially the change of order that is performed after a first computation of a Gr\"obner basis for a suitable order
	to deduce a basis for the lexicographic order which is more suited for outputting a solution) it is helpful 
	to use additional equations that restrict the solution space which is otherwise really huge in this case.
	The purpose of the following discussion is to show where these solutions come from and how to restrict them. 	
	By bilinearity of System \eqref{syst:eqRkSgn} we may fix 
	\begin{equation} \label{eq:rkSgnGrob1} 
	x_1=1
	\end{equation} 
	when there is a solution $\xv$ such that $x_1 \neq 0$).
	Furthermore,  the fact that $\Cpub'$ is an $\Fq$ vector space of dimension $n/d$ induces 
	that for a given $\xv$ solution to  \eqref{syst:eqRkSgn} the set of corresponding $\yv$'s also forms 
	a vector space of dimension $n/d$. We may therefore rather safely assume that we can choose 
	\begin{equation} \label{eq:rkSgnGrob2}
	\forall i \in \llbracket 1,\frac{n}{d} - 1 \rrbracket, \mbox{ } y_{i} = 0 \quad \mbox{and} \quad y_{n/d} = 1. 
	\end{equation} 
	There is an additional degree of freedom on $\xv$ coming from the fact that even if $d=2$ there are several 
	spaces $\alpha F$ for which $1 \in \alpha F$. To verify this, let us study in more detail the case when 
	$F$ is of dimension $2$, say
	$$
	F = \vsg{a,b}.
	$$
	We wish to understand what are the possible values for $z \in \Fqm$ such that there exists $c \neq 0$ 
	for which $$
	\vsg{a,b} = c \vsg{1,z}.$$
	The possible values for $\xv$ will then be the projection of those $z$ to the $\Fq$ space $\vsg{\beta_1,\dots,\beta_{m-1}}$. The possible values for $z$ are then obtained from studying the possible values for $c$. There are two 
	cases to consider:
	
	\begin{itemize} 
		\item {\bf Case 1:} $c = \frac{\mu}{a + b \nu}$ for $\mu \in \Fq^\times$ and $\nu \in \Fq$. In such a case $$z = \frac{\beta b}{a + b \nu} + \delta$$ for $\beta \in \Fq^\times$, $\delta \in \Fq$.
		\item {\bf Case 2:} $c = \frac{\mu}{b}$ for $\mu \in \Fq^\times$. Here
		$$
		z = \alpha \frac{a}{b} + \delta
		$$
		for $\alpha \in \Fq^\times$, $\delta \in \Fq$.
	\end{itemize} 
	Since the $\delta$ term vanishes after projecting $x$ onto $\vsg{\beta_1,\dots,\beta_{m-1}}$ we have essentially two degrees of freedom over $\Fq$ for $x$. One has already been taken into account when setting $x_1=1$. We can
	add a second one $x_2 = \alpha$ where $\alpha$ is arbitrary in $\Fq$. We have actually chosen in our experiments that 
	\begin{equation} \label{eq:rkSgnGrob3}
	(x_2-\alpha)(x_2 - \beta)=0 
	\end{equation} 
	for some random $\alpha$ and $\beta$ in $\Fq$. This has resulted
	in some gain in the computation of the solution space. Finally the following proposition summarizes the system we have solved.

	\begin{proposition} By eliminating variables using Equations \eqref{eq:rkSgnGrob1},\eqref{eq:rkSgnGrob2} and \eqref{eq:rkSgnGrob3} in \eqref{syst:eqRkSgn} we have
		\begin{itemize}
			\item $	nm - k(m+1) - t + 2$ equations;

			\item $m-1 + n+t$ unknowns.
		\end{itemize} 
		
	\end{proposition}

	In the ``typical regime'' where $m \approx n$, $k \approx \frac{n}{2}$ and $t \ll n$ 
	we have a number of equations of order $n^2$ and a number of unknowns of order $n$, therefore 
	typically the regime where we expect that the Gr\"obner basis techniques take polynomial time.

	\subsection{Numerical results}

	We give in Table \ref{table:numRes} our numerical results to find a codeword of rank $2$ in any public code of the RankSign scheme for parameters chosen according to \cite{AGHRZ17}. These results have been obtained with an Intel Core i5 processor, clocked at $1.6$ GHz using a single core, with $8$ Go of RAM.

	\begin{table} 
		\begin{center} 
			\begin{tabular}{|c|c|c|c|}
				\hline
				Intended Security \cite{AGHRZ17} & $(n,k,m,d,t,q)$ & Time & Maximum Memory Usage \\
				\hline		
				128 bits & $(20,10,21,2,2,2^{32})$ & $20.12$ s & 49 MB \\
				128 bits& $(24,12,24,2,2,2^{24})$ & $31.75$ s & 65 MB \\
				192 bits& $(24,12,27,2,3,2^{32})$ & $125.64$ s & 97 MB \\
				256 bits & $(28,14,30,2,3,2^{32})$ & $256.90$ s & 137 MB \\
				\hline 			
			\end{tabular} 
		\end{center} 
		\caption{Attack on NIST's parameters of RankSign \label{table:numRes} }
	\end{table}

	\subsection{Finishing the attack}\label{ss:finishing}

	\noindent We present in this subsection the end of our attack which consists in being able to sign with only the knowledge of the public key. It holds for the parameters chosen for the NIST competition \cite{AGHRZ17} for which 
	$d=2$. Observe that \eqref{eq:problem} implies that we have $k = n-k=n/2$.
	
	We have at that point obtained the code $\Cpub'$ (see \S\ref{subsec:outlineAtt}, Point 3.) that has dimension (over $\Fq$) $\geq n/d=n/2=k$.	
	This code is just $\Fq$-linear, but it will be convenient to extend it by considering its 
	$\Fqm$-linear extension, that we denote $\Fqm \otimes \Cpub'$ that is defined by 
	the $\Fqm$-linear subspace of $\Fqm^{n+t}$ obtained from linear combinations over $\Fqm$ of codewords in 
	$\Cpub'$. In other words if we denote by $\{\cv'_1,\dots,\cv'_{k'}\}$ an $\Fq$-basis 
	of $\Cpub'$, then
	$$
	\Fqm \otimes \Cpub' = \vsgm{\cv'_1,\dots,\cv'_{k'}}.
	$$
	
	To simplify the discussion we make now  the following assumption (which was 
	corroborated by our experiments)
	\newline

	\begin{ass} \label{ass:1} 
		$$
		\dim \Fqm \otimes \cC_{\textup{pub}}' = k.
		$$
	\end{ass}

	\noindent The rationale behind this assumption is that (i) the dimension of $\Cpub'$ is very likely to be $n/d$ which is equal to $k$ and
	(ii) an $\Fq$ basis of $\Cpub'$ is very likely to be an $\Fqm$ basis too.

	\begin{lemma}
		Under Assumption \ref{ass:1} the code $\left(\Fqm \otimes \cC_{\textup{pub}}'\right)^{\bot}$ has length $n+t$, dimension $n+t-k$ and is an \textup{LRPC}-code that is associated to a homogeneous matrix that has all its entries in an $\Fq$ subspace $F$ of $\Fqm$ of dimension $2$ which contains $1$. Furthermore, the sets 
		$$
		\cD \eqdef \{ \cv \in \left(\Fqm \otimes \cC_{\textup{pub}}'\right)^{\bot} : \Sp(\cv) \subseteq \Fq \}
		\mbox{ and }   \cD' \eqdef \{ \cv \in \left(\Fqm \otimes \cC_{\textup{pub}}'\right)^{\bot} : \Sp(\cv) \subseteq F \} 		$$
		are $\Fq$-subspaces of dimension  $\geq t$
		and $\geq n-k+2t$ respectively.

	\end{lemma}

	\begin{proof} By Assumption \ref{ass:1}, $\Fqm \otimes \Cpub'$ is of dimension $k$ and its dual
		$\left(\Fqm \otimes \cC_{\textup{pub}}'\right)^{\bot}$ has therefore dimension $n+t-k$.
		There is a generator matrix for $\Fqm \otimes \Cpub'$ that is formed by rows taken from $\Cpub'$.
		It is homogeneous of weight $2$. Say that its entries generate a space $F$.
		This is also a parity-check matrix of the dual code. $\left(\Fqm \otimes \cC_{\textup{pub}}'\right)^{\bot}$
		is therefore an LRPC code of weight $2$.
		
		By applying now Lemma \ref{lem:pcp} to it with $F' = \Fq$, we have
		\begin{align*} 
		(n+t) \dim_{\Fq}(\Fq) - (n+t - (n+t-k)) \dim(F \cdot \Fq) &= n + t - 2k \\
		&= t \mbox{ (because } 2k = n)
		\end{align*}  
		which gives the result for the set $\cD$. We apply once again Lemma \ref{lem:pcp} but this time with $F'=F$.
		Say $F  = \langle 1,x_{1} \rangle_{\Fq}$. This  gives a lower bound on the dimension of $\cD'$ which is
		\begin{align*} 
		(n+t) \dim (F) - (n+t - (n+t-k))\dim(F\cdot F) &\geq 2(n+t) - 3k \\
		&\mbox{ (because } F\cdot F = \langle 1,x_{1},x_{1}^{2} \rangle_{\Fq}) \\ 
		&= n-k + 2t \mbox{ (because }2k = n).
		\end{align*}

	\end{proof}

	\noindent To end our attack we make now the following assumption that was again corroborated in our experiments.
	\newline

	\begin{ass} \label{ass:2} We can extract from sets $\cD$ and $\cD'$ a basis of $\left(\Fqm \otimes \cC_{\textup{pub}}'\right)^{\bot}$ with 
		\begin{enumerate}
			\item $t$ codewords of support $\Fq$,

			\item $n-k$ codewords of a same support of rank $2$ which contains $1$.
		\end{enumerate}
	\end{ass}

	\begin{restatable}{lemma}{lemfinishing}
		\label{lem:finishing} Under Assumptions \ref{ass:1} and \ref{ass:2} there exists a parity-check matrix $\Hm' \in \Fqm^{(n+t-k)\times (n+t)}$ of $\Fqm \otimes \cC_{\textup{pub}}'$, an invertible matrix $\Pm$ of size $n+t$ 
		with entries in the small field $\Fq$ and an invertible matrix $\Sm$ of size $n+t-k$ with entries  in $\Fqm$ such that
		$$
		\Sm \Hm' \Pm = 
		\begin{pmatrix}
		I_{t} & \mathbf{0} \\
		\mathbf{0} & \Rm 
		\end{pmatrix}
		$$
		where $\Rm$ is homogeneous of degree $2$ and of size $(n-k) \times n$.

	\end{restatable} 
	
	\begin{proof}

		Under Assumptions \ref{ass:1} and \ref{ass:2} there is a generator matrix of $\left(\Fqm \otimes \cC_{\textup{pub}}'\right)^{\bot}$ and thus a parity-check matrix of $\Fqm \otimes \cC_{\textup{pub}}'$ which is homogeneous of degree $2$ with the particularity that $t$ rows of it are of rank $1$. Let $\Hm'$ be such a matrix, thus by making a Gaussian elimination on its rows we have an invertible matrix $\Sm \in \F_{q^m}^{(n+t-k)\times (n+t-k)}$ such that:
		$$
		\Sm\Hm' = \begin{pmatrix}
		I_{t} & \Qm \\
		\mathbf{0} & \Rm 
		\end{pmatrix}
		$$ 
		where $\Qm$ is a matrix of size $t\times (n+t)$ whose entries lie in the small field $\Fq$
		and $\Rm$ is a homogeneous matrix of weight $2$ and of size $(n-k)\times n$. In this way there exists an invertible matrix $\Pm$ of size $(n+t)\times (n+t)$ with coefficients in the field $\Fq$ such that 
		$$
		\Sm \Hm' \Pm = \begin{pmatrix}
		I_{t} & \mathbf{0} \\
		\mathbf{0} & \Rm 
		\end{pmatrix}
		$$
		which concludes the proof. 
	\end{proof}

	\noindent The idea now to sign as a legitimate user will be to use the matrix $\Rm$ and the decoder of Fact 
	\ref{fa:LRPCdecoder} (see Section \S\ref{sec:rkSgn}).
	Recall that to make a signature for the matrix $\Hpub$ (which defines the public code $\cC_{\textup{pub}}$) and a message $\mv$, we look for an error $\ev$ of rank $w$ satisfying
	$n-k  = d(w-t-t')$ (see  Equation \eqref{eq:w} of Fact \ref{fa:LRPCdecoder}),
	such that $\Hpub\transp{\ev} = \transp{\sv}$ with $\sv = \cH(\mv)$ (the hash of the message). The algorithm that follows performs this task:
	\bigskip

	\framebox{\begin{minipage}{0.9\textwidth}
			
			\begin{enumerate}
				\item We compute $\yv \in \Fqm^{n+t}$ such that $\Hpub\transpose{\yv} = \transpose{\sv}$.

				\item Let $\yv' = \yv \transpose{(\Pm^{-1})}$ and we compute $\sv' = (\Sm\Hm'\Pm)\transpose{\yv'}$.

				\item Let $\sv_{1}'$ be the first $t$ coordinates of $\sv'$, $\sv'_{2}$ its last $n-k$ ones. We apply the decoder  of \S\ref{sec:rkSgn} with:

				\qquad - The subspace $T \eqdef \Sp(\sv_{1}') + T'$ where $T'$ is a random subspace of $\Fqm$ 
				of dimension $t'$.

				\qquad - The parity-check matrix $\Rm$ and the syndrome $\sv_{2}'$.

				\noindent Then we get a vector $\ev'$ such that $T \subseteq \Sp(\ev')$ and $\Rm\transpose{\ev'} = \transpose{\sv_{2}'}$.

				\item We compute $\ev = (\sv_{1}',\ev')\transpose{\Pm}$.
			\end{enumerate}
		\end{minipage}
	}$$\mbox{} $$
	
	Let us now show  the correctness of this algorithm, in other words we show that $\Hpub\transp{\ev} = \transp{\sv}$ with $|\ev| = w$ satisfying $n-k  = 2(w-t-t')$.
	\begin{proof}[Proof of Correctness]
		First we have:
		\begin{align*}
		\transpose{\Hm'\ev} &= \Hm'\Pm\transpose{(\sv_{1}',\ev')} \\
		&= \Sm^{-1} (\Sm\Hm'\Pm)\transpose{(\sv_{1}',\ev')} \\
		&= \Sm^{-1} \begin{pmatrix}
		I_{t} & \mathbf{0} \\
		\mathbf{0} & \Rm 
		\end{pmatrix} \transpose{(\sv_{1}',\ev')} \\
		&= \Sm^{-1} \begin{pmatrix} \transpose{\sv_{1}'} \\
		\Rm \transpose{\ev'} 
		\end{pmatrix} \mbox{ (because } \sv_{1}' \mbox{ of size } t)\\
		&= \Sm^{-1} \transpose{\sv'} \mbox{ (because } \Rm\transpose{\ev'} = \transpose{\sv_{2}'}) \\
		&= \Sm^{-1} (\Sm \Hm'\Pm) \transpose{\yv'} \\
		&= \Hm' \Pm (\Pm^{-1}\transpose{\yv}) \\
		&= \Hm'\transpose{\yv}\\
		\end{align*}
		which implies that $\Hm'\transpose{(\ev - \yv)} = \mathbf{0}$ and $\yv - \ev \in \Fqm \otimes \cC_{\textup{pub}}'$. Recall now that $\Fqm \otimes \cC_{\textup{pub}}' \subseteq \cC_{\textup{pub}}$ and therefore $\Hpub\transpose{(\ev - \yv)} = \mathbf{0}$. By  linearity we get $\Hpub\transpose{\ev} = \Hpub\transpose{\yv} = \transpose{\sv}$.

		\noindent Thus under the condition that the decoder in Point 3 works for the matrix $\Hm'$, the syndrome $\sv$ and the subspace $T$,  our algorithm decodes the syndrome $\sv$ relatively to $\Hpub$.

		\noindent The parity-check matrix $\Rm$ is homogeneous of degree $2$, has $n-k$ rows and $n$ columns.
		We can therefore apply to it the decoder of \S\ref{sec:rkSgn}. 
		It will output (we use here Fact \ref{fa:LRPCdecoder}) an error $\ev'$ of weight $w'$ that satisfies
		$n-k=2(w'-t-t')$. Note that this implies that $w'=w$ which is the error weight we want to achieve.
		Then the error $\ev = (\sv_{1}',\ev')\transpose{\Pm}$ has the same rank as $\Sp(\sv_{1}') \subseteq T \subseteq \Sp(\ev')$ and $\Pm$ is an invertible matrix in the small field which concludes the proof.  
		
	\end{proof}

	\section{Attack on the IBE in the rank metric} \label{subsec:attRSL} 
	
	In the previous section we showed that RankSign is not a secure signature scheme. This also shows the insecurity of the IBE proposal made in \cite{GHPT17a} since it is 
	partly based on it. It could be thought that it just suffices to replace in the IBE scheme \cite{GHPT17a} RankSign by another signature scheme in the rank metric. 
	This is already problematic, since RankSign was the only known rank metric code-based signature scheme up to now. We will actually show here 
	that the problem is deeper than this. We namely show that the parameters proposed in \cite{GHPT17a} can be broken by an algebraic attack that attacks the RSL problem directly and not the
	underlying signature scheme. We will however show that the constraints on the parameters of the scheme coming from
	Proposition \ref{propo:paramIBE} together with the new constraint for avoiding the algebraic attack exposed here can in theory be met.
	In the IBE \cite{GHPT17a} we are given a matrix $\Gm_{\cC_{\sgn}}$ of size $k_{\sgn} \times n_{\sgn}$ whose coefficients live in $\Fqm$ and the matrix $\Gm_{\cC_{\sgn}}\Em$ where $\Em$ has size $n_{\sgn} \times n_{\dec}$ with all its coefficients which live in a same secret subspace $F$ of dimension $w_{\dec}$ and an attacker wants  to recover $F$. We show in \S\ref{subsec:attRSL} 
	that under the condition $n_{\dec} > w_{\dec}(n_{\sgn} - k_{\sgn})$ (which is verified in \cite{GHPT17a}) the code $\cC$ defined by
	\begin{equation}\label{eq:C}
	\cC = \{ \ev\transpose{(\Gm_{\textup{sgn}}\Em)} : \ev \in \Fq^{n_{\dec}} \} \subseteq \Fqm^{k_{\textup{sgn}}}. 
	\end{equation}
	is an $\Fq$-subspace which contains words of weight $\leq w_{\dec} $ which reveal $F$. 
	It turns out that the subspace $\cC' \eqdef \cC \cap F^{k_{\sgn}}$ of words of $\cC$ whose coordinates all live in $F$ is of  dimension $\geq n_{\dec} - w_{\dec} (n_{\sgn} - k_{\sgn})$.
	We then apply standard algebraic techniques in Subsection \S\ref{subsec:attRSLgrob} to recover $\cC'$ and therefore $F$ from it. This breaks all the parameters proposed in  \cite{GHPT17a}. 
	We conclude this section by showing that there is in principle a way to choose the parameters of the IBE scheme to possibly avoid 
	this attack.

	\subsection{Low rank codewords from instances of the RSL problem}\label{subsec:attRSL}

	We prove here that a certain $\Fq$-linear code  that contains many low-weight codewords can be computed by the attacker. This is explained by
	\begin{theorem} \label{theo:attRSL}
		Let $(\Am,\Am\Em)$ be an instance of \textup{RSL} for parameters $n,k,N,w $ with $\Am\in \Fqm^{(n-k)\times n}$ in systematic form and $\Em \in \Fqm^{n\times N}$ where all its coefficients belong to a same subspace $F$ of dimension $w$. Furthermore, we suppose that
		\begin{equation} \label{eq:consRSL}
		N > w k.
		\end{equation}
		Let
		\begin{eqnarray*}
			\cC &\eqdef &\{ \ev\transpose{(\Am \Em)} : \ev \in \Fq^{N} \} \\ 
			\cC' & \eqdef& \cC \cap  F^{n-k}.
		\end{eqnarray*}
		$\cC'$ is an $\Fq$-subspace of $\cC$ of dimension $\geq N - wk$.
	\end{theorem}

	\begin{proof} Let us first decompose $\Em$ in two parts $\left\lbrack\frac{\Em_{1}}{\Em_{2}}\right\rbrack$ where $\Em_{1}$ is formed by the first $n-k$ rows of $\Em$ and $\Em_{2}$ by the last $k$ ones. The matrix $\Am$ is in systematic form, namely $(I_{n-k}|\Am')$ where $\Am' \in \Fqm^{(n-k) \times k}$, which gives:
		$$
		\Am\Em = \Em_{1} + \Am'\Em_{2} 
		$$
		Therefore, to prove our theorem we just need to show that 
		$$
		\cS \eqdef \{\ev \in \Fq^{N} : \Em_{2}\transpose{\ev} = \mathbf{0} \}
		$$ 
		is an $\Fq$-subspace of dimension greater than $N-wk$. Indeed, for each error $\ev$ of $\cS$ we have $(\Am\Em)\transpose{\ev} = \Em_{1}\transpose{\ev}$ which belongs to $F^{n-k}$ as coefficients of $\Em_{1}$ are in the $\Fq$-subspace $F$ and those of $\ev$ are in $\Fq$.

		\noindent Denote the entry in row $i$ and column $j$ of $\Em_{2}$ by $E_{i,j}$. A word of $\cS$ satisfies
		$$
		\forall i \in \IInt{1}{k},  \quad \sum_{j=1}^{N} E_{i,j}e_{j} = 0.
		$$
		Looking in addition for $\ev$ that has all its entries in $\Fq$ and expressing these $k$ linear equations 
		over $\Fqm$ in a basis of $F$ (since 
		$\sum_{j=1}^{N} E_{i,j}e_{j}$ belongs by definition to $F\cdot \Fq = F$) we obtain $k \dim(F)=kw$ linear equations over $\Fq$ involving $N$ unknowns (the $e_{j}$'s) in $\Fq$.
		The solution space is therefore of dimension greater than $N -wk$ which concludes the proof of the theorem.  
	\end{proof}

	\subsection{How to find low rank codewords in instances of the RSL problem} \label{subsec:attRSLgrob} 
	Theorem \ref{theo:attRSL} showed that there are many codewords of weight $\leq w_{\dec} $ in the code $\cC$ defined in \eqref{eq:C}.
	Let us show now how these codewords can be recovered by an algebraic attack. The sufficient condition $n_{\dec} > w_{\dec}(n_{\textup{sgn}} - k_{\textup{sgn}})$
	ensuring the existence of such codewords is met for the parameters proposed in \cite{GHPT17a}.

	To explain our algebraic modeling of the problem, let us first recall that for a fixed basis $(\beta_{1},\cdots,\beta_{m})$ of $\Fqm$ over $\Fq$ we can view elements of $\Fqm^{k_{\sgn}}$ as matrices of size $m \times k_{\sgn}$:
	$$
	\forall \xv \in \Fqm^{k_{\sgn}}, \quad \Mat{\xv} = (X_{i,j}) \in \Fq^{m \times k_{\sgn}}\mbox{ where } x_{j} = \sum_{i=1}^{m} \beta_{i}X_{i,j}.
	$$
	The associated matrix code $\cC^{\textup{Mat}}$ is defined as:
	$$
	\cC^{\textup{Mat}} \eqdef \{ \Mat{\cv} : \cv \in \cC \} \subseteq \Fq^{m \times k_{\sgn}}.
	$$
	
	\noindent It is easily verified that this matrix-code has dimension $n_{\dec}$. It is clear now by applying Theorem \ref{theo:attRSL} that:
	\newline

	\noindent {\bf Fact 3.} $\cC^{\textup{Mat}}$ contains codewords of rank $\leq \dim(F)$ which form a $\Fq$-subspace of dimension $\geq n_{\dec} - w_{\dec} (n_{\sgn} - k_{\sgn})$.
	\newline

	These are just the codewords $\cv'$ which are of the form $\Mat{\cv}$ where $\cv \in \cC$ with $\Sp(\cv) \subseteq F$. 
	We do not expect other codewords of this rank in $\cC^{\textup{Mat}}$ since $w_{\dec} $ is much smaller than 
	the Varshamov-Gilbert bound in the case of the parameters proposed in \cite{GHPT17a}.
	\newline

	\noindent {\bf The basic bilinear system.} Finding codewords of rank $w_{\dec} $ in $\cC^{\textup{Mat}}$ can be expressed as an instance of the MinRank problem \cite{BFS99,C01}. Once again we propose the algebraic modeling which was suggested in \cite{AGHRZ17}. It consists here in setting up the algebraic system with unknowns $\xv^{i} = (x^{i}_{1},\cdots,x^{i}_{m}) \in \Fq^{m}$ and $\yv^{i}_{j} \in \Fq^{k_{\sgn}}$ for $1 \leq i \leq w_{\dec} $ and $1 \leq j \leq k_{\sgn}$ where the $\xv^{i}$'s can be thought as a basis of the unknown subspace $F$ and the $\yv_{j}^{i}$'s as coordinates of the codeword in this basis. In that case the codeword $\Mm$ of $\cC^{\textup{Mat}}$ of rank $w_{\dec} $ has the following form:
	$$
	\Mm = \begin{pmatrix} \sum_{i=1}^{w_{\dec} } x_1^{i} y_1^{i} & \sum_{i=1}^{w_{\dec} }x_1^{i} y_2^{i} & \hdots & \sum_{i=1}^{w_{\dec} }x_1^{i} y_{k_{\sgn}}^{i} \\
	\sum_{i=1}^{w_{\dec} }x_2^{i} y_1^{i} & \sum_{i=1}^{w_{\dec} }x_2^{i} y_2^{i} & \hdots & \sum_{i=1}^{w_{\dec} }x_2^{i} y_{k_{\sgn}}^{i} \\
	\vdots & \vdots & \vdots  & \vdots \\
	\sum_{i=1}^{w_{\dec} }x_{m}^{i} y_1^{i} & \sum_{i=1}^{w_{\dec} }x_{m}^{i} y_2^{i} & \hdots & \sum_{i=1}^{w_{\dec} }x_{m}^{i} y_{k_{\sgn}}^{i}
	\end{pmatrix}.
	$$

	\noindent Recall now that $\Cmat$ has the structure of an $\Fq$-subspace of $\Fq^{m\times k_{\sgn}}$ of dimension
	$n_{\dec}$. By viewing the elements of $\Cmat$ as vectors of $\Fq^{mk_{\sgn}}$, i.e. the matrix $\Mm=(M_{ij})_{\substack{1 \leq i \leq m\\ 1 \leq j \leq k_{\sgn}}}$ is viewed as the vector $\mv=(m_\ell)_{1 \leq \ell \leq mk_{\sgn}}$ where
	$m_{(i-1)k_{\sgn}+j} = M_{i,j}$, we can compute a parity-check matrix $\Hm^{\textup{Mat}}$ for it. It is an 
	$(mk_{\sgn} - n_{\dec})\times mk_{\sgn}$ matrix that we denote by $\Hm^{\textup{Mat}} = (H^{\textup{Mat}}_{ij})_{\substack{1 \leq i \leq mk_{\sgn}- n_{\dec} \\ 1 \leq j \leq mk_{\sgn} }}$. This matrix gives $mk_{\sgn}- n_{\dec}$ bilinear equations that have
	to be satisfied by the $x^{l}_{i}$'s and the $y^{l}_{j}$'s:
	
	\begin{equation} \label{syst:eqIBE} 
	\left\{
	\begin{array}{l}
	\mathop{\sum}\limits_{l=1}^{w_{\dec} }\mathop{\sum}\limits_{j = 1}^{k_{\sgn}} \mathop{\sum}\limits_{i=1}^{m} H^{\textup{Mat}}_{ 1,(i-1)k_{\sgn} + j} x_{i}^{l} y_{j}^{l} = 0 \\
	\qquad\vdots \\ 
	\mathop{\sum}\limits_{l=1}^{w_{\dec} }\mathop{\sum}\limits_{j = 1}^{k_{\sgn}} \mathop{\sum}\limits_{i=1}^{m} H^{\textup{Mat}}_{ mk_{\sgn} - n_{\dec},(i-1)k_{\sgn} + j} x_{i}^{l} y_{j}^{l} = 0 \\
	\end{array}
	\right.
	\end{equation}

	\noindent {\bf Restricting the number of solutions.}
	We have solved the bilinear system \eqref{syst:eqIBE} with Gr\"obner basis techniques that are implemented in Magma. To speed-up the resolution, as in the case of the attack on RankSign, we add new equations to \eqref{syst:eqIBE} which come from the vectorial structure of $F$ and the set of solutions.

	\noindent With our notation we can view $F$ as an $\Fq$ subspace of $\Fq^{m}$ of dimension $w_{\dec}$ generated by the rows of the matrix:
	$$
	\begin{pmatrix} 
	x_{1}^{1} & \cdots & x_{m}^{1} \\
	x_{1}^{2} & \cdots & x_{m}^{2} \\ 
	\vdots & & \vdots \\ 
	x_{1}^{w_{\dec} } & \cdots & x_{m}^{w_{\dec} }
	\end{pmatrix} 
	$$
	In this way, we can put this matrix into systematic form, it will generate the same subspace. Therefore 
	we can add equations 
	\begin{equation} \label{eq:IBE1} 
	\forall (i,j) \in \llbracket 1,w_{\dec}  \rrbracket^{2}, \mbox{ } j \neq i, \quad x_{i}^{j} = 0 \mbox{ and }  x_{i}^{i} = 1
	\end{equation} 
	without modifying the set of codewords of rank $w_{\dec} $. Furthermore, this set is an $\Fq$-subspace of dimension greater than $n_{\dec} - (n_{\sgn} - k_{\sgn})w_{\dec} $ and as in the case of the attack on RankSign we 
	may assume that for a random subset $I \subseteq \llbracket 1,k_{\sgn} \rrbracket \times \llbracket 1,w_{\dec}  \rrbracket$ of size $n_{\dec} - (n_{\sgn} - k_{\sgn}) - 1$ there is an element in this set for which:
	\begin{equation} \label{eq:IBE2} 
	\forall (j,i) \in I, \mbox{ } y^{i}_{j} = 0 \mbox{ and } y^{i_{0}}_{j_{0}} = 1 \mbox{ for } (i_{0},j_{0}) \notin I.
	\end{equation} 
	Equations \eqref{eq:IBE1} and \eqref{eq:IBE2} enable us to reduce the number of variables of the previous bilinear system. The following proposition summarizes the number of equations and variables that we finally get.

	\begin{proposition} By eliminating variables using Equations \eqref{eq:IBE1} and \eqref{eq:IBE2} in \eqref{syst:eqIBE} we obtain
		\begin{itemize}
			\item $mk_{\sgn} + w_{\dec} ^{2} + (n_{\sgn} - k_{\sgn})$ equations;

			\item $mw_{\dec}  + k_{\sgn}w_{\dec} $ unknowns.
		\end{itemize} 
		
	\end{proposition}
	In the ``typical regime'' where $m \approx n_{\sgn} \approx k_{\sgn}$ and $w_{\dec}  \approx n_{\sgn}^\varepsilon$ for 
	some $\varepsilon$ in $(0,1)$ we have a number of equations of order $n_{\sgn}^2$ and a number of unknowns of order $n_{\sgn}^{1+\varepsilon}$, therefore 
	typically the regime where we expect that the Gr\"obner basis techniques take subexponential time.

	\subsection{Numerical results}

	We give in Table \ref{table:numResIBE} our numerical results to find codewords of rank $w_{\dec} $ in instances of the RSL problem for the parameters chosen according to \cite{GHPT17a}. These results have been obtained with an Intel Core i5 processor, clocked at $1.6$ GHz using a single core, with $8$ Go of RAM. In our implementation, we verified that when we generated an instance whose 
	associated secret is the subspace $F$ we only got codewords whose coordinates live in this subspace  and therefore revealed it. 
	
	\begin{table} 
		\begin{center} 
			\begin{tabular}{|c|c|c|c|}
				\hline
				Intended Security & $(n_{\sgn},k_{\sgn},m,w_{\dec},n_{\dec},k_{\dec},q)$ & Time & Maximum Memory Usage \\
				\hline		
				128 bits & $(100,80,96,4,96,9,2^{192})$ & $626$s & $1.7$ GB \\
				\hline 			
			\end{tabular} 
		\end{center} 
		\caption{Attack on parameters of the rank-based IBE \cite{GHPT17a} \label{table:numResIBE} }
	\end{table}

	\subsection{Avoiding the attack}

	Although our attack breaks the parameters proposed in \cite{GHPT17a}, there might in principle be a way to instantiate the IBE 
	with a new signature scheme. Recall that the constraints that have to be satisfied are given by 
	\begin{eqnarray}
	\wrVG(q,m,n_{\sgn},k_{\sgn}) & \leq w_{\sgn} \leq & \frac{m(n_{\sgn}-k_{\sgn})}{\max(m,n_{\sgn})} \text{ (signature constraint)}  \label{eq:signature}\\
	w_{\sgn} w_{\dec}  & \leq & \wrVG(q,m,n_{\dec},k_{\dec}) \text{ (decoding works)} \label{eq:decoding}\\
	w_{\dec} (n_{\sgn} - k_{\sgn}) & \geq & n_{\dec} \text{ (for avoiding our attack)}. \label{eq:constraint}
	\end{eqnarray} 
	
	The lower-bound in \eqref{eq:signature} ensures that we can find a signature whereas the role of the upper-bound is to ensure that the 
	problem of finding a signature does not become easy. The constraint \eqref{eq:decoding} is here to ensure that the decoding procedure 
	used for recovering the plaintext works and the last constraint is here to avoid our attack. This set of parameters is non-empty under the condition to find an efficient hash and sign signature scheme. For instance, if we have a signature scheme which achieves the lower bound \eqref{eq:signature}, namely $w_{\sgn} = \wrVG(q,m,n_{\sgn},k_{\sgn})$ we can choose:
	$$
	n_{\sgn} = 100 \quad ; \quad k_{\sgn} = 75 \quad ; \quad n_{\dec} = 96 \quad ; \quad k_{\dec} = 4 \quad ; \quad w_{\dec}  = 4.
	$$

	More generally, if one wants to set parameters of the IBE \cite{GHPT17a} we propose to proceed in the following way.
	We first propose to choose $m = n_{\sgn}$ and a signature code for which the ratio $\frac{\wrVG(q,m,n_{\sgn},k_{\sgn})}{n_{\sgn}-k_{\sgn}}$ is sufficiently small (it can even approach $\frac{1}{2}$) and
	we choose 
	\begin{equation}
	\label{eq:choice1}
	w_{\sgn} = (1 - \varepsilon) (n_{\sgn}-k_{\sgn})
	\end{equation}
	for some appropriate $\varepsilon$. We then choose an $\Fqm$-linear code of parameters $[n_{\dec},k_{\dec}]$ of sufficiently small
	dimension such that 
	$$
	\wrVG(q,m,n_{\dec},k_{\dec}) \geq (1-\varepsilon) n_{\dec}.
	$$
	This is possible in principle. Therefore we can choose $w_{\dec} $ such that $w_{\sgn}w_{\dec}  \geq (1-\varepsilon)n_{\dec}$ and for which
	\eqref{eq:decoding} holds. By satisfying the two first constraints \eqref{eq:signature} and \eqref{eq:decoding} in this way, we also satisfy the last one, namely Equation \eqref{eq:constraint}. 
	This can be verified by arguing that 
	\begin{eqnarray*}
		w_{\dec} (n_{\sgn} - k_{\sgn}) &=  & \frac{w_{\sgn} w_{\dec} }{1-\varepsilon} \;\;\text{ (we use \eqref{eq:choice1})}\\
		& \geq & \frac{n_{\dec}(1-\varepsilon) }{1-\varepsilon} \;\;\text{ (we use the particular choice of $w_{\dec} $)}\\
		& = & n_{\dec}
	\end{eqnarray*}

	\subsection{Comparison with previous attacks against RSL}

	Recall here that the Rank Support Learning (RSL) problem for parameters $n,k,N,w $ can also be expressed as follows: we have access to a matrix of full rank $\Am\in \Fqm^{(n-k) \times n}$ and to $N$ syndromes $\Am\transpose{\ev}$ for $\ev$ chosen uniformly at random in $F^n$ where $F$ is some fixed subspace of $\Fqm$ of dimension $w $.
	The problem is then to recover $F$. When $N = 1$ this is just the Rank Syndrome Decoding (RSD) problem (see Problem \ref{def:RSD} in \S\ref{subsec:rankBasedCrypto}). It is readily verified that the difficulty of RSL decreases when $N$ grows, however the question for cryptographic purposes is: ``how large $N$ can be while RSL remains hard?'' In \cite[\S4, p14]{GHPT17a} a first answer was given by showing that $N$ has to verify 
	\begin{equation} \label{eq:consN}
	N < w n
	\end{equation}
	otherwise a polynomial attack can easily be mounted. 
	Here, we strengthen  this condition on $N$, we require namely that in order to avoid our new attack we should have	
	\begin{equation}\label{eq:consOurs}
	N \leq w k \quad (k <n).
	\end{equation}
	where $k$ is the dimension of the code of parity-check matrix $\Am \in \Fqm^{(n-k)\times n}$ used in the instance of RSL. 
	This condition is clearly stronger since we always have $k < n$ at the cost of trading a polynomial attack in the case where \eqref{eq:consN} is met with 
	a subexponential attack when \eqref{eq:consOurs} is not met.
	
	In the context of the IBE it is actually {\em significantly stronger}.
	This comes from the fact that in this context we really expect that under reasonable assumptions that $k \ll n$. This can be explained as follows.
	In this case \eqref{eq:consOurs} translates into $n_{\dec} \leq w_{\dec} (n_{\sgn} - k_{\sgn})$. The point is that in the typical regime which is needed for the IBE, 
	we have $n_{\sgn}-k_{\sgn} \ll n_{\sgn}$. By typical regime we mean here that we can assume that for the IBE \cite{GHPT17a} we have
	
	\begin{ass}\label{ass:0}
		\begin{equation} \label{ass:00}
		\frac{k_{\sgn}}{n_{\sgn}} = \Omega(1);
		\end{equation}
		\begin{equation}\label{ass:10} 
		m = \Theta(n_{\sgn});
		\end{equation}
		\begin{equation}\label{ass:01}
		w_{\dec}=O(1)
		\end{equation}
	\end{ass} 
	This assumption is minimal in the special case of the IBE \cite{GHPT17a} as we are going to explain.

	Equations \eqref{ass:00} and \eqref{ass:10} ensure that we are in the regime where the Gilbert-Varshamov and the Singleton bounds do not collapse which is essential as explained in the previous subsection to obtain  parameters avoiding our attack on RSL.

	Equation \eqref{ass:01} permits to avoid a polynomial attack against the problem RSL. Indeed, suppose  that $w_{\dec}$ is bounded, which is $w_{\dec} = O(1)$. Recall that in the IBE, instances of RSL have the following form $(\Gm_{\cC_{\sgn}},\Gm_{\sgn}\Em)$ where $\Gm_{\cC_{\sgn}} \in \Fqm^{k_{\sgn} \times n_{\sgn}}$ and $\Em$ is homogeneous with underlying subspace $F$ of dimension $w_{\dec}$. Solving here the Rank Syndrome Decoding for $\Gm_{\cC_{\sgn}}$, a weight $w_{\dec}$ and the first column of $\Em$ as syndrome will give with  high probability $F$ as $w_{\dec}$ is smaller than the Varshamov-Gilbert bound. By using Gr\"obner basis techniques for this and writing equations in the small field $\Fq$ this gives: 
	\begin{enumerate}
		\item $n_{\sgn}w_{\dec}$ unknowns;

		\item $mk_{\sgn}$ bilinear equations.
	\end{enumerate}
	Under Assumptions \eqref{ass:00},\eqref{ass:10} and the fact that $k_{\sgn}\leq n_{\sgn}$, we have $mk_{\sgn} = \Theta(n_{\sgn}k_{\sgn}) = \Theta(n_{\sgn}^{2})$. On the other hand, the number of unknowns is $O(n_{\sgn})$ as $w_{\dec} = O(1)$. This is the regimewhere we expect  to solve the corresponding bilinear system in polynomial time.
	Therefore we can safely assume that  $w_{\dec}$ tends to infinity to avoid such a  polynomial attack.

	Assumption \ref{ass:0} leads in this case to the following proposition
	\begin{proposition}
		Under Assumption \ref{ass:0}, we have when $n_{\sgn}$ tends to infinity:
		$$
		n_{\sgn} - k_{\sgn} = o(n_{\sgn}).
		$$ 
	\end{proposition}

	\begin{proof} From Proposition \ref{propo:paramIBE} we have:
		$$
		w_{\sgn}w_{\dec} \leq \min(n_{\dec},m)
		$$
		and thus from Assumption \ref{ass:0}:
		$$
		w_{\sgn}w_{\dec} = O(n_{\sgn})
		$$
		which gives 
		\begin{equation} \label{eq:d} 
		\frac{w_{\sgn}}{n_{\sgn}} = O\left( \frac{1}{w_{\dec}} \right).
		\end{equation}

		Under Assumption \ref{ass:0} we have that $w_{\dec}$ tends to infinity. Therefore we get
		\begin{equation} \label{eq:wsgn}
		w_{\sgn} = o(n_{\sgn})
		\end{equation}
		Now under the signature constraint (see \eqref{eq:signature}) we have that $w_{\sgn} \geq \wrVG(q,m,n_{\sgn},k_{\sgn})$. From Equation \eqref{eq:VGdist} and Assumption \ref{ass:0}, particularly \eqref{ass:10}, it is easily verified that the last inequality and \eqref{eq:wsgn} imply $n_{\sgn} - k_{\sgn} = o(n_{\sgn})$ which concludes the proof of the proposition. 	
		
	\end{proof}

	\section{Attack on the IBE in the Hamming metric}
	
	The purpose of this section is to show that there is an even more fundamental problem with the general IBE scheme given in Section \ref{sec:IBEscheme} in the Hamming metric.
	We will namely prove here that due to the constraint on the parameters coming from Proposition \ref{propo:paramIBE}, we can not find a set of parameters which would avoid
	an attack based on using generic decoding techniques. Even the simplest of those techniques, namely the Prange algorithm \cite{P62}, breaks the IBE in the Hamming metric in polynomial time. 
	We refer the reader to Section \S\ref{sec:IBEscheme} where we introduced all the notations that we are going to use.

	To show that the IBE can be attacked in the Hamming metric we proceed as follows. The attacker knows 
	$\Gm_{\cC_{\sgn}}\Em$ and that the columns of $\Em$ have weight $w_{\dec}$. We will show that we can solve efficiently for  the range of parameters admissible for the IBE the following syndrome decoding problem: given a matrix $\Gm_{\cC_{\sgn}} \in \F_{2}^{k_{\sgn}\times n_{\sgn}}$ and $\sv \in \F_{2}^{k_{\sgn}}$ such that there exists $\ev \in \F_{2}^{n_{\sgn}}$ of weight $w_{\dec}$ for which 
	$\Gm_{\cC_{\sgn}}\transpose{\ev} = \transpose{\sv}$, we want to recover $\ev$. This allows to recover the columns of $\Em$ and therefore $\Em$.
	The scheme is broken with this knowledge, since the attacker also knows 
	$\cH(id)\Em + \mv\Gm_{\cC_{\dec}}$, $\cH(id)$ and $\Gm_{\cC_{\dec}}$. This is used to derive $\mv \Gm_{\cC_{\dec}}$ and finally $\mv$.
	
	To solve this decoding problem, we use the Prange algorithm (see \cite{P62}) whose complexity is, up to a polynomial factor in $n_{\sgn}$, equal to:
	\begin{equation} \label{eq:cpxPrange} 
	\frac{ \binom{n_{\sgn}}{w_{\dec}} }{\binom{k_{\sgn}}{w_{\dec}}}
	\end{equation} 
	In the special case of the IBE we proved in Proposition \ref{propo:paramIBE} that the parameters have to verify the following constraint:
	$$
	w_{\sgn}w_{\dec} = O(n_{\sgn}).
	$$
	Now the parameter $w_{\sgn}$ can not be too small either, since for fixed $w_{\sgn}$ the algorithms for decoding linear codes also solve the signature forgery in polynomial time.
	This problem amounts in the case of the IBE to find a $\uv_{id}$ such that 
	$$
	|\uv_{id}\Gm_{\cC_{\sgn}} - \cH(id)| = w_{\sgn}.
	$$
	We will therefore make a minimal assumption that ensures that the decoding algorithms for solving this problem have at least 
	some (small) subexponential complexity. We also make the same assumption for the aforementioned recovery of $\ev$. This is obtained by assuming that
	\begin{ass}\label{ass:minimal}$ $
		\begin{eqnarray} 
		w_{\dec} & = &\Omega(n^{\varepsilon}) \quad \mbox{ for some } \quad \varepsilon > 0 \label{eq:ass2}\\
		w_{\sgn} & = & \Omega(n^{\varepsilon'}) \quad \mbox{ for some } \quad \varepsilon' > 0  \label{eq:ass3}
		\end{eqnarray} 	
	\end{ass}

	\begin{restatable}{proposition}{proppolynomial}
		\label{prop:polynomial}
		Under Assumption \ref{ass:minimal}, the Prange algorithm breaks the IBE scheme in Hamming metric in  polynomial time in $n_{\sgn}$. 
	\end{restatable}
	
	\begin{proof} 	Recall that parameters of the IBE in Hamming metric are $n_{\sgn},k_{\sgn},w_{\sgn}$. For the sake of simplicity let,
		$$
		n \eqdef n_{\sgn} \quad ; \quad k \eqdef k_{\sgn} \quad ; \quad w \eqdef w_{\sgn}. 
		$$
		We start the proof by noticing that $ww_{\dec} = O(n)$ and  Assumption \ref{ass:minimal} actually imply the   ``converse'' inequalities
		\begin{eqnarray*}
			w_{\dec} & = &O(n^{1-\varepsilon'}) \\
			w_{\sgn} & = & O(n^{1-\varepsilon}) 
		\end{eqnarray*}

		Let us now derive an asymptotic expression for \eqref{eq:cpxPrange} by studying:
		$$
		\log_{2}\binom{n}{w_{\dec}}  - \log_{2} \binom{k}{w_{\dec}} 
		$$
		This is obtained through:
		\begin{lemma} \label{lemm:binent}
			$$
			\log_{2} \binom{n}{l}  = nh\left( \frac{l}{n} \right) + O\left(\log_{2}n \right) 
			$$
			for $h$ being defined over $\lbrack 0,1 \rbrack$ as:
			$$
			h(x) \eqdef -x \log_{2}x - (1-x)\log_{2}x 
			$$
		\end{lemma}
		\noindent Therefore to show that the Prange algorithm is polynomial in $n$ it is sufficient to prove that:
		\begin{equation}\label{eq:pcpx} 
		n h \left( \frac{w_{\dec}}{n} \right) - k h \left( \frac{w_{\dec}}{k} \right)
		\end{equation} 
		is an $O\left(\log_{2} n \right)$. 
		Recall now that in a context of code-based hash and sign, the weight of the decoding has to be greater than the Varshamov Gilbert bound, namely:
		$$
		w \geq nh^{-1} \left( 1 - \frac{k}{n} \right) \iff \frac{k}{n} \geq 1 - h \left( \frac{w}{n} \right) 
		$$

		\noindent From $w = O\left(n^{1-\varepsilon}\right)$, $k \leq n$  and 
		by using that $h(x) \mathop{=}\limits_{x \rightarrow 0} O\left(- x \log_{2}x \right)$ we get:
		\begin{equation}\label{eq:1proof}
		\frac{k}{n} =  1  +  O\left( - \frac{w}{n} \log_{2}  \frac{w}{n} \right)
		\end{equation}
		and $k\sim n$.
		We are now ready to show that \eqref{eq:pcpx} is asymptotically an $O(\log_{2}n)$. As $k \sim n$ and $w_{\dec} = O(n^{1-\varepsilon'})$  by using now that $h(x) \mathop{=}\limits_{x \rightarrow 0} -x \log_{2}x + \frac{1}{\ln(2)} \left( x -\sum_{l=2}^{p-1}\frac{x^{l}}{l(l-1)} + O(x^{p})\right)$ for an integer $p$ greater than $\frac{1}{\varepsilon'}$ we have:
		\begin{equation} \label{eq:cpxprr}
		n h \left( \frac{w_{\dec}}{n} \right) - k h \left( \frac{w_{\dec}}{k} \right) = a(n) + b(n) + c(n)
		\end{equation} 
		where 
		$$
		a(n) \eqdef  k  \frac{ w_{\dec} }{ k } \log_{2} \frac{ w_{\dec} }{ k } -n \frac{w_{\dec}}{n}\log_{2}\frac{w_{\dec}}{n}
		$$
		\begin{equation} \label{eq:b(n)}
		b(n) \eqdef   n\frac{w_{\dec}}{\ln(2)n} - k \frac{w_{\dec}}{\ln(2)k} + \frac{1}{\ln(2)} \sum_{l=2}^{p-1} k\frac{(w_{\dec}/k)^{l}}{l(l-1)} - n \frac{(w_{\dec}/n)^{l}}{l(l-1)}   
		\end{equation}
		and
		$$
		c(n) \eqdef  n O\left( \frac{w_{\dec}^{p}}{n^{p}} \right) + kO\left( \frac{w_{\dec}^{p}}{k^{p}} \right)  
		$$ 
		We easily have:
		\begin{align*} 
		c(n) &= O \left( \frac{w_{\dec}^{p}}{n^{p-1}} \right) \\
		&= O \left( \frac{1}{n^{p\varepsilon' -1}} \right) \quad (\mbox{because } w_{\dec} = O(n^{1-\varepsilon'})) \\
		&= o(1) \quad (\mbox{because } p> 1/\varepsilon') 
		\end{align*} 
		Let us now compute $a(n)$:
		\begin{align*} 
		a(n)&= k  \frac{ w_{\dec} }{ k } \log_{2} \frac{ w_{\dec} }{ k } - n \frac{w_{\dec}}{n}\log_{2}\frac{w_{\dec}}{n} \\
		&=  w_{\dec} \log_{2} \frac{w_{\dec}}{k} - w_{\dec} \log_{2} \frac{w_{\dec}}{n} \\
		&= -w_{\dec} \log_{2} \frac{k}{n} \\
		&= -w_{\dec} \log_{2} \left(  1  +  O\left( \frac{w}{n} \log_{2} \frac{w}{n} \right) \right) \quad \mbox{(because of \eqref{eq:1proof})} 
		\end{align*}
		Recall now that we have $w = O(n^{1-\varepsilon})$ and by using $\log_{2}(1+x) \mathop{=}\limits_{x \rightarrow 0} O(x) $ we get:
		\begin{align*} 
		a(n) &= -w_{\dec} O\left( \frac{w}{n} \log_{2} \frac{w}{n}  \right)   
		\end{align*} 
		which gives as $ww_{\dec} = O(n)$ and $w = O(n^{1-\varepsilon})$:
		\begin{equation} \label{eq:resa(n)}
		a(n) =  O\left( \log_{2}  \frac{w}{n} \right) = O\left( \log_{2}n \right) 
		\end{equation} 
		
		\noindent Let us now compute $b(n)$ which is defined in \eqref{eq:b(n)}:
		\begin{align*} 
		b(n) &=   n\frac{w_{\dec}}{\ln(2)n} - k \frac{w_{\dec}}{\ln(2)k} + \frac{1}{\ln(2)} \sum_{l=2}^{p-1} k\frac{(w_{\dec}/k)^{l}}{l(l-1)} - n \frac{(w_{\dec}/n)^{l}}{l(l-1)} \\
		&= \frac{1}{\ln(2)} \sum_{l=2}^{p-1} w_{\dec}^{l}\frac{1}{n^{l-1}}\frac{ (k/n)^{1-l} - 1}{l(l-1)} 
		\end{align*} 
		Recall now that 
		$$
		\frac{k}{n}=  1  +  O\left(- \frac{w}{n} \log_{2} \frac{w}{n} \right)
		$$
		therefore,
		\begin{align*} 
		\left(\frac{k}{n}\right)^{1-l} -1 &= \left(1  +  O\left(- \frac{w}{n} \log_{2} \frac{w}{n}   \right) \right)^{1-l} - 1 \\
		&= 1  + O\left(- \frac{w}{n}\log_{2}\frac{w}{n} \right)  - 1 \\
		&=  O\left(- \frac{w}{n} \log_{2} \frac{w}{n}  \right) 
		\end{align*} 
		
		\noindent Then we get:
		\begin{align*}
		b(n) &= \frac{1}{\ln(2)}  \sum_{l=2}^{p-1} \frac{w_{\dec}^{l}}{n^{l-1}}\frac{1}{l(l-1)}   O\left(- \frac{ w}{n} \log_{2}\frac{w}{n} \right) \\
		&= \frac{1}{\ln(2)}  \sum_{l=2}^{p-1} \frac{w_{\dec}^{l-1}}{l(l-1)n^{l-1}} O \left(- \frac{w_{\dec}w}{n} \log_{2}  \frac{w}{n}  \right)    \\
		&= \frac{1}{\ln(2)}  \sum_{l=2}^{p-1} o(1)O\left(- \log_{2} \frac{w}{n} \right) \\
		& \qquad\qquad \mbox{(because } w_{\dec}  = O(n^{1-\varepsilon'}) = o(n) \mbox{ as } \varepsilon'>0 \mbox{ and } w_{\dec} w = O(n)). \\
		&= o\left(-\log_{2} \frac{w}{n} 	\right)   \\
		&= O\left(\log_{2}n \right)
		\end{align*}

		\noindent Therefore, by combining this result with \eqref{eq:cpxprr} and \eqref{eq:resa(n)} we get:
		$$
		kh\left( \frac{ w_{\dec} }{k} \right) - nh\left( \frac{w_{\dec} }{n} \right) = O\left(\log_{2}n\right)
		$$
		which concludes the proof. 
		
	\end{proof} 
	
	\section{Concluding remarks}\label{sec:conclusion}
	
	We have presented here our attacks against the rank-based signature scheme RankSign and the IBE scheme proposed in \cite{GHPT17a}. 
	Several comments can be made.
	\newline

	\noindent {\bf Attack on RankSign.} 
	We actually showed that in the case of RankSign, the complexity is polynomial for all possible strategies for choosing the parameters. 
	Repairing the RankSign scheme seems to require to modify the scheme itself, not just adjust the parameters.
	It might be tempting to conjecture that the approach  against RankSign could also be used to mount an attack on the NIST submissions
	based on LRPC codes such as \cite{ABDGHRTZ17,ABDGHRTZ17a}.  Roughly speaking our approach consists in looking for low weight codewords in the LRPC code instead of looking for low weight codewords in the usual suspect, that is the dual of the LRPC code, that has in this case low weight codewords by definition of the LRPC code. This approach 
	does not seem to carry over to the LRPC codes considered in those submissions. The point is that our approach was 
	successful for RankSign because of the way the parameters of the LRPC code had to be chosen. In particular the length $n$,
	the dimension $k$ and the weight of the LRPC code have to satisfy
	$$
	n = (n-k) d.
	$$
	It is precisely this equality that is responsible for the weight $2$ codewords in the LRPC code. If $d$ is not too small
	(say $>3$) and $(n-k)d$ is sufficiently above  $n$, then the whole approach considered here fails at the very beginning. 
	\newline

	\noindent {\bf Attack on the IBE \cite{GHPT17a}.} 
	The attack on RankSign also breaks the IBE proposal of \cite{GHPT17a} since it is based partly on the RankSign primitive. We have shown here that the problem is actually deeper
	than this by showing that even if a secure signature scheme replaces in the IBE, RankSign, then an attack that breaks directly the RSL problem which is the other 
	problem on which the IBE is based, can be mounted for the parameters proposed in \cite{GHPT17a}. Again, as in the case of RankSign, the reason why this attack was successful comes from 
	the fact that the constraints on the parameters that are necessary for the scheme to work properly work in favor of ensuring that a certain code that can be computed from 
	the public data has low weight codewords. These low codewords are then found by an algebraic attack. However, contrarily to the RankSign case, where the conditions on the parameters
	force a certain code to have codewords of low weight, this phenomenon can be avoided by a very careful choice of the parameters in the IBE. This opens the way for repairing the scheme
	of \cite{GHPT17a} if a secure signature scheme is found for the rank metric.
	
	We have also studied whether the \cite{GHPT17a} approach for obtaining an IBE scheme based on coding assumptions could work in the Hamming metric. However in this case, and contrarily to what happens in the rank metric, we have given a devastating polynomial attack in the Hamming metric relying on using the simplest 
	generic decoding algorithm \cite{P62} that can not be avoided by any reasonable choice of parameters. 
	It seems that following the GPV \cite{GPV08}/\cite{GHPT17a} approach for obtaining an IBE scheme is a dead end in the case of the Hamming metric.

	To conclude this discussion on \cite{GHPT17a},
	we would like to stress that our result in the Hamming case does not imply the impossibility of designing an IBE based on coding theory. It only suggests to investigate other paradigms rather 
	than trying to adapt the GPV strategy. For instance, the recent progress of \cite{DG17,DG17a,DGHM18} made on the design of IBE's,  particularly with the concept of 
	one-time signatures with encryption, might be applied to cryptography based on  decoding assumptions.

	\newcommand{\etalchar}[1]{$^{#1}$}

\end{document}